\documentclass[11pt,a4paper]{amsart}
\usepackage[top=4cm,bottom=4cm,left=3.5cm,right=3.5cm]{geometry}
 
\usepackage{subcaption}
\usepackage{tabularx}
\usepackage{amsmath,amssymb,amsthm}
\usepackage[english]{babel}
\usepackage{mathtools}
\usepackage{enumerate}
\usepackage[numbers]{natbib}
\usepackage[hidelinks]{hyperref}
 
\usepackage{comment}
 
\usepackage{multirow}
\usepackage[dvipsnames]{xcolor}
\usepackage{algorithm}
\usepackage[noend]{algpseudocode}
\usepackage[colorinlistoftodos]{todonotes}

\usepackage{nicematrix}

\DeclareMathOperator{\rk}{rk}

\DeclareMathOperator{\supp}{supp}

\DeclareMathOperator{\wt}{wt}

\newtheorem{theorem}{Theorem}

\newtheorem{lemma}[theorem]{Lemma}
\newtheorem{proposition}[theorem]{Proposition}
\newtheorem{corollary}[theorem]{Corollary}
\theoremstyle{definition}

\newtheorem{definition}[theorem]{Definition}
\theoremstyle{remark}

\newtheorem{example}[theorem]{Example}

\numberwithin{equation}{section}

\newcommand{\Z}{\mathbb{Z}}

\newcommand{\Zp}[1]{\Z/{#1}\Z}

\newcommand{\field}{\mathbb{F}}
\newcommand{\Fq}{\mathbb{F}_q}

\newcommand{\zps}{\mathbb{Z} /p^s \mathbb{Z}}
\newcommand{\zpsk}[1]{\left(\mathbb{Z}/p^s\mathbb{Z}\right)^{#1}}

\newcommand{\card}[1]{\left\vert \, {#1} \,\right\vert} 
\newcommand{\floor}[1]{\left\lfloor\, {#1}\, \right\rfloor} 
\newcommand{\set}[1]{\left\lbrace{#1}\right\rbrace} 

\newcommand{\st}{\, | \,} 
\newcommand{\tendsto}{\longrightarrow}

\newcommand{\prob}{\mathbb{P}} 

\newcommand{\code}{\mathcal{C}}
\newcommand{\subcode}{\mathcal{D}}

\newcommand{\LW}{\wt_{\scriptsize\mathsf{L}}}
\newcommand{\LWjoin}{\wt_{\scriptsize\mathsf{L, join}}}
\newcommand{\LWmeet}{\wt_{\scriptsize\mathsf{L, meet}}}

\newcommand{\HW}{\wt_{\scriptsize\mathsf{H}}}
\newcommand{\HWjoin}{\wt_{\scriptsize\mathsf{H, join}}}
\newcommand{\HWmeet}{\wt_{\scriptsize\mathsf{H, meet}}}
\DeclareMathOperator{\dist}{d}
\newcommand{\HD}{\dist_{\scriptsize\mathsf{H}}}
\newcommand{\LD}{\dist_{\scriptsize\mathsf{L}}}
\newcommand{\LDjoin}{\dist_{\scriptsize\mathsf{L, join}}}
\newcommand{\LDmeet}{\dist_{\scriptsize\mathsf{L, meet}}}

\newcommand{\lweight}{\wt_\mathsf{L}}

\newcommand{\suppjoin}{\supp_{\mathsf{ join}}}
\newcommand{\suppmeet}{\supp_{\mathsf{meet}}}
\newcommand{\suppH}{\supp_{\mathsf{H}}}
\newcommand{\suppHjoin}{\supp_{\mathsf{H, join}}}
\newcommand{\suppHmeet}{\supp_{\mathsf{H, meet}}}
\newcommand{\suppL}{\supp_{\mathsf{L}}}
\newcommand{\suppLC}{\supp_{\mathsf{L, col}}}
\newcommand{\suppLjoin}{\supp_{\mathsf{L, join}}}
\newcommand{\suppLmeet}{\supp_{\mathsf{L, meet}}}

\newcommand{\wtC}{\wt_{\mathsf{C}}}
\newcommand{\wtLC}{\wt_{\mathsf{L, col}}}
\newcommand{\distC}{\dist_{\mathsf{C}}}
\newcommand{\distLC}{\dist_{\mathsf{L, col}}}

\newcommand{\Gsys}{G_{\mathsf{sys}}}
\newcommand{\Grsys}{G_{\mathsf{rsys}}}

\usepackage{stmaryrd}

\newcommand{\commentOut}[1]{}

\colorlet{myblue}{cyan!55!teal}
\colorlet{myred}{red!30!purple}
\colorlet{mygreen}{green!70!teal}

\colorlet{jgreen}{PineGreen!60!LimeGreen}
\colorlet{Jgreen}{teal!80!jgreen}

\newcommand{\highlightred}[1]{\colorbox{WildStrawberry!70}{#1}}

 \usepackage{ytableau}
 \ytableausetup
{mathmode, boxframe=normal, boxsize=3.5em}

\def\zpsn{\left(\mathbb{Z} /{p^s}\mathbb{Z}\right)^n}

\def\supp{{\mathrm{supp}}}

\usepackage{longtable}
\usepackage{multirow,multicol}
\usepackage{caption}
\usepackage{ltablex}\keepXColumns

\newcolumntype{Y}{>{\centering\arraybackslash}X}
\newcolumntype{Z}{>{\scriptsize}Y}
\usepackage{rotating}

\usepackage{tikz,pgfplots}
\pgfplotsset{compat=newest}

\title{Better bounds on the minimal Lee distance}
\begin{document}
\author[J. Bariffi]{Jessica Bariffi$^{1,2}$}
\author[V. Weger]{Violetta Weger$^3$}

\address{$^1$Institute of Communication and Navigation \\ 
        German Aerospace Center \\ 
        Germany
\newline \indent $^2$Institute of Mathematics \\ 
        University of Zurich \\ 
        Switzerland 
}
\email{jessica.bariffi@dlr.de}
\address{$^3$Department of Electrical and Computer Engineering\\
	Technical University of Munich\\
        Germany
	}
\email{violetta.weger@tum.de}

\maketitle
\begin{abstract}
   This paper provides new and improved Singleton-like bounds for Lee metric codes over integer residue rings. We derive the bounds using  various novel definitions of generalized Lee weights based on different notions of a support of a linear code. In this regard, we introduce three main different support types for codes in the Lee metric and analyze their utility to derive bounds on the minimum Lee distance. Eventually, we propose a new point of view to generalized weights and give an improved bound on the minimum distance of codes in the Lee metric for which we discuss the density of maximum Lee distance codes with respect to this novel Singleton-like bound.
\end{abstract}


\section{Introduction}\label{sec:intro}
The Lee metric was introduced in 1958 by Lee \cite{lee1958some} to cope with phase modulation in communication. It provides an interesting alternative to the Hamming and rank metric which are considered for orthogonal modulation and network coding, respectively. The Lee metric is most known for the celebrated result in \cite{hammons1994z} where the authors showed that some optimal non-linear binary codes can be represented as linear codes over $\mathbb{Z}/4\mathbb{Z}$ endowed with the Lee metric.
Recently, there is a renewed interest in the Lee metric due to its similarities with the Euclidean norm used in lattice-based cryptography. In fact, the Lee metric was introduced to cryptography in \cite{leeZ4} and its hard problems were further studied in \cite{LeeNP}. Recently, a first Lee-metric primitive \cite{fuleeca} has been submitted to the reopened NIST standardization process for digital signature schemes.

Although the Lee metric is one of the oldest metrics and has interesting properties and applications, it did not receive as much attention as  other metrics which is visible in the lacking of a well understood algebraic foundation of Lee-metric codes.  Indeed, only recently it was discovered that Lee-metric codes attain the Gilbert-Varshamov bound with high probability \cite{free}, for the length of the code tending to infinity. This aligns with famous and well studied results in the Hamming \cite{forney} and the rank metric \cite{loidreau}.
In addition, the characterization of constant Lee weight codes, initiated by Wood \cite{wood}, has only recently been completed in \cite{byrne2023bounds}.

The study of optimal codes is a classical direction 
in algebraic coding theory and the most famous bound within this direction is the Singleton bound. 
The bound gives an upper bound  on the minimum distance of a code, given other parameters of the code. Thus, codes attaining this bound have the maximal possible minimum distance and in turn the largest error correction capacity.

In the Hamming metric, the Singleton bound was introduced by Singleton \cite{singleton} and was already studied by Komamiya \cite{komamiya}. 
Codes in the Hamming metric that are attaining the Singleton bound are called Maximum Distance Separable (MDS) codes and are constituting some of the most used and studied codes in coding theory. It is well known that codes attain the Hamming-metric Singleton bound with high probability when letting the size of the finite field $q$ tend to infinity. 
Instead, if one lets the length $n$ tend to infinity, the probability for a code to attain the Singleton bound tends to 0.

When changing the underlying metric, this behaviour can drastically change. This is for example the case in the rank metric. 
While the rank metric is younger than the Lee metric, first being introduced by Delsarte \cite{delsarte} in 1978 and reintroduced by Gabidulin \cite{gabidulin} and Roth \cite{roth}, its Singleton bound was already given in \cite{gabidulin} in 1985 and its optimal codes, called Maximum Rank Distance (MRD) codes, have been well studied since. In fact, we know that for $\mathbb{F}_{q^m}$-linear codes, MRD codes are dense when letting $q$ or $m$ tend to infinity \cite{ale}. For $\mathbb{F}_q$ linear codes, however MRD codes are sparse when letting $q$ tend to infinity \cite{anina}, except for some special cases, where $m$ or $n$ are 2 \cite{heide2, BR19, heide}.

The situation in the Lee metric is completely different. 
Indeed, the first Singleton bound in the Lee metric has only been stated in 2000 by Shiromoto \cite{shiromoto}. This bound is tight, as there exists a code attaining it. 
 However, the recent paper \cite{byrne2023bounds}
revealed that this example is in fact the only non-trivial linear code that is optimal with respect to this Lee-metric Singleton bound. Thus, such optimal codes are  "extremely" sparse and show the need of a more thorough study of bounds in the Lee metric and their optimal codes. 
The used puncturing argument to derive the classical Singleton bound and thus also the Lee-metric alternative in \cite{shiromoto} proves to be not suitable for the Lee metric. One thus requires other techniques to derive a Singleton-like bound which appear to are more tailored to the Lee metric setting. In fact, one  possible technique is through generalized weights, first introduced in \cite{gen}.

In this paper we introduce several possible definitions of Lee supports of a code, which allows us  to define generalized Lee weights. 
We compare the resulting Lee-metric Singleton bounds and compute the density of their optimal codes.
The main contribution of this paper is a new Lee-metric Singleton bound for which the optimal codes are not sparse.

This paper is structured as follows. In Section \ref{sec:prelim} we provide the preliminary basics on ring-linear codes and the Lee metric, which will be helpful for the remainder of the paper. 
We also provide the already known Singleton-like bounds for the Lee metric and the densities of their optimal codes. 
Section \ref{sec:genWeight_subcodes} serves as a recap of generalized Hamming weights for finite fields and for finite integer rings. We restate there the main properties and definitions and discuss how to derive bounds on the minimum Lee distance using generalized weights. In Sections \ref{sec:joinsupport} and \ref{sec:colsupport} we introduce two new support definitions in the Lee metric deriving from the ideas used in the Hamming metric. For both of them we are able to derive new bounds on the minimum Lee distance defining generalized Lee weights with respect to the new supports. Additionally, we show that the newly proposed generalized Lee weights are invariant under isometries in the Lee metric, and we discuss the density of optimal codes. Even though these new bounds are sharper than the existing bounds on the minimum Lee distance, they are still not sufficiently tight. Therefore, we introduce new generalized weights based on filtrations of a code in Section \ref{sec:genWeight_filtration}. Together with some additional parameters for a generator matrix of a filtration of a code we are able to derive an improved bound on the minimum Lee distance. Again, we discuss the density of optimal codes and the invariance under isometries for the new generalized weights. In Section \ref{sec:comparison} we compare all the bounds for several parameters. 
Finally, we conclude the paper in Section \ref{sec:conclusions}.

\section{Preliminaries}\label{sec:prelim}
We start this section by introducing the main definitions, notions and results used in the course of this paper. Throughout this paper we denote by $p$ a 
prime number, by $s$ a positive integer, and by $\zps$ the integer residue ring. Furthermore, for any integer $i \in \set{0, \ldots, s-1}$, we write $\langle p^i \rangle$ to denote either the ideal $p^i(\zps)$ or the submodule $p^i(\zps)^n$. 

\subsection{Ring-Linear Codes}\hfill\\
Classical coding theory considers finite fields $\mathbb{F}_q$ with $q$ elements and a linear code $\mathcal{C}$ is a subspace of the vector space $\mathbb{F}_q^n$. Thus, $\mathcal{C}$ has a dimension $k$, which determines its size $\mid \mathcal{C} \mid =q^k$ and the minimum number of generators.
Such a code $\mathcal{C}$ can then be represented through a generator matrix $G \in \mathbb{F}_q^{k \times n}$ or a parity-check matrix $H \in \mathbb{F}_q^{(n-k) \times n}$, which have the code as image, respectively as kernel.

In this paper, we focus on codes over $\mathbb{Z}/p^s\mathbb{Z}.$

\begin{definition}\label{def:code}
    A linear \textit{code} $\code \subseteq (\zps)^n$ is a $\zps$-submodule of $(\zps)^n$.
\end{definition}

Due to the fundamental theorem of finite abelian groups, a linear code is isomorphic to the following direct sum of $\mathbb{Z}/p^s\mathbb{Z}$-modules
$$\mathcal{C} \cong \bigoplus\limits_{i=1}^K \left(\mathbb{Z}/p^s\mathbb{Z} \right) / \langle p \rangle^{\lambda_i}.$$

The type of a module $\mathcal{C}$ is then defined as the partition
$$\lambda= (\underbrace{s, \ldots, s}_{k_0}, \underbrace{s-1, \ldots, s-1}_{k_1}, \ldots, \underbrace{1, \ldots, 1}_{k_{s-1}})$$ or equivalently
$\lambda=(s^{k_0}(s-1)^{k_1}\cdots 1^{k_{s-1}})$. Instead of using this well-known notation from the theory of  modules, we prefer the  notation $(k_0, \ldots, k_s)$, called \emph{subtype} of the code $\mathcal{C}$, due to its simplicity.

 We call $n$ the \textit{length} of the code $\code$ and its elements \textit{codewords}. We can define the $\mathbb{Z}/p^s\mathbb{Z}$-dimension of a linear code $\mathcal{C}$ as
  \begin{align*}
        k := \log_{p^s}(\card{\code}).
    \end{align*}
The \textit{rank} $K$ of a code is given by the number of generators, i.e., $K := \sum_{i = 0}^{s-1} k_i$. Additionally,  $k_0$ is called the \textit{free rank} of the code $\code$. 
We denote by the $\mathbb{Z}/p^s\mathbb{Z}$-\emph{dimension} of the code $\mathcal{C}$ the following 
$$k= \log_{p^s}\left(|\mathcal{C}|\right).$$
Note that the $\mathbb{Z}/p^s\mathbb{Z}$-dimension of a code is not necessarily an integer. In fact, $k$ is determined by the subtype as $$k= \sum\limits_{i=0}^{s-1} \frac{s-i}{s}k_i.$$
In general it holds that $0\leq k_0 \leq k \leq K \leq n.$
If the rank and the $\mathbb{Z}/p^s\mathbb{Z}$-dimension of a code coincide, we call the code \emph{free}. 
It has been shown in \cite{free} that free codes are dense as $p \longrightarrow \infty$ but they are neither dense nor sparse as the  length $n$ or $s$ tend to infinity.


As in the finite field case, ring-linear codes are represented through a generator matrix or a parity-check matrix. 
\begin{definition}\label{def:GH}
    Consider a code $\code \subseteq (\zps)^n$ of rank $K$. A matrix $G \in (\zps)^{K \times n}$ is called a generator matrix of $\code$ if the rows of $G$ span the code. A \textit{parity-check matrix} $H$ is an $(n-K)\times n$ matrix over $\zps$ whose its null-space coincides with $\code$.
\end{definition}
Usually it is helpful to consider these matrices in their systematic form.
\begin{proposition}\label{prop:systematicG}
    Let $\code$ be a linear code in $\zpsk{n}$ of subtype $(k_0, \ldots, k_{s-1})$ and rank $K$. Then $ \code$ is permutation equivalent to a code having a generator matrix $\Gsys \in \zpsk{K \times n}$ of the form
        \begin{equation}\label{systematicformG}
            \Gsys =
            \begin{pmatrix}
                \mathbb{I}_{k_0}  & A_{1,2} & A_{1,3} & \cdots & A_{1,s} & A_{1,s+1}  \\
                0 & p\mathbb{I}_{k_1} &  pA_{2,3} & \cdots & pA_{2,s} & pA_{2,s+1}  \\
                0 & 0 & p^2 \mathbb{I}_{k_2} & \cdots & p^2A_{3,s}  & p^2 A_{3,s+1} \\
                \vdots & \vdots &  \vdots & & \vdots & \vdots \\
                0 & 0 & 0  & \cdots &  p^{s-1} \mathbb{I}_{k_{s-1}} & p^{s-1} A_{s,s+1} 
            \end{pmatrix},
        \end{equation}
    where $A_{i,s+1} \in (\mathbb{Z} / p^{s+1-i}\mathbb{Z})^{k_{i-1} \times (n-K)}, A_{i,j} \in (\mathbb{Z} / p^{s+1-i} \mathbb{Z})^{k_{i-1} \times k_j}$ for $j \leq s$.
    In addition, the code $\code$ is permutation equivalent to a code having a parity-check matrix $H \in \zpsk{(n-k_0) \times n}$ of the form
    \begin{equation}\label{systematicformH}
        H_{\mathsf{sys}} =
        \begin{pmatrix}
            B_{1,1} & B_{1,2} & \cdots & B_{1,s-1} & B_{1,s} & \mathbb{I}_{n-K} \\
            pB_{2,1} & pB_{2,2} & \cdots & pB_{2,s-1} & p\mathbb{I}_{k_{s-1}} & 0  \\
            p^2B_{3,1} & p^2B_{3,2} & \cdots & p^2 \mathbb{I}_{k_{s-2}} & 0 & 0 \\
            \vdots & \vdots & & \vdots & \vdots & \vdots \\
            p^{s-1}B_{s,1} & p^{s-1}\mathbb{I}_{k_1} & \cdots &  0 & 0 & 0
        \end{pmatrix},
    \end{equation}
    where $B_{1,j} \in (\mathbb{Z} / p^s\mathbb{Z})^{(n-K) \times k_{j+1}}, B_{i,j} \in (\mathbb{Z} / p^{s+1-i} \mathbb{Z})^{k_{s-i+1} \times k_{j+1}}$ for $i > 1$.
\end{proposition}
We call the forms in \eqref{systematicformG} and \eqref{systematicformH} the \textit{systematic form} of a generator matrix and a parity-check matrix, respectively.

Additionally to the subtype of a code $\code \subseteq (\zps)^n$, we can define a similar parameter going over the columns of a generator matrix of $\code$.
\begin{definition}\label{def:supp_subtype}
    Let $\code \subseteq (\zps)^n$ be a linear code of rank $K$. For each $j \in \set{1, \ldots , n}$ consider the $j$-th coordinate map
    \begin{align*}
        \begin{array}{rccl}
            \pi_j: & \zpsn & \longrightarrow & \zps \\
                    & (c_1,\ldots,c_n) & \longmapsto & c_j
        \end{array}.
    \end{align*}
    The \textit{support subtype} of $\code$ is defined to by an $(s+1)$-tuple $(n_0(\code), \ldots, n_s(\code))$, where $n_i(\code)$ counts the number coordinates $j \in \set{1, \ldots ,n}$ belonging to ideal $\langle p^i \rangle$, i.e.,
    \begin{align*}
        n_i(\code) := \card{\set{ j \in \set{1, \ldots, n} \st \langle \pi_j(\code) \rangle = \langle p^i \rangle }}.
    \end{align*}
    A code with $n_s(\code) = 0$ is called \textit{non-degenerate}.
\end{definition}
We will simply write $n_i$ instead of $n_i(\code)$ if the code $\code$ is clear from the context. 


\subsection{Lee Metric}\hfill\\
In this paper we will focus on the Lee metric introduced in \cite{lee1958some}. This metric was introduced to cope with phase modulation in communications. 
However, we will often refer and compare the Lee metric to the Hamming metric. 

Thus, recall that for two $n$-tuples $x, y \in (\zps)^n$ their \textit{Hamming distance} is defined to be the number of positions where they differ, i.e.,
\begin{align*}
    \HD(x, y) = \card{\set{i \in \set{1, \ldots, n} \st x_i \neq y_i}}.
\end{align*}
For a single element $x\in(\zps)^n$ its \textit{Hamming weight} is the Hamming difference to zero, i.e., the number of nonzero entries of $x$, and we denote it by $\HW(x) := \HD(x, 0)$. The \textit{minimum Hamming distance} of a linear code $\code \subseteq (\zps)^n$ is then defined as
\begin{align*}
    \dist_{\mathsf{H}}(\code) := \min\set{ \HW(x) \st x \in \code\setminus\set{0}}.
\end{align*}
Let us now introduce the Lee metric.
\begin{definition}\label{def:leemetric}
    Given an integer residue ring $\zps$ and consider an element $a \in \zps$ interpreted as an integer in $\set{0, \ldots , p^s-1}$. The \textit{Lee weight} $a\in\zps$ is given by
    \begin{align*}
        \LW(a) := \min\set{a, \card{p^s-a}}.
    \end{align*}
    For $x \in (\zps)^n$ its Lee weight is defined in an additive fashion. That is,
    \begin{align*}
        \LW(x) = \sum_{i = 1}^n \LW(x_i).
    \end{align*}
\end{definition}
Similarly to the Hamming metric, the Lee weight induces a distance which we refer to as \textit{Lee distance}. In fact, for $x, y \in (\zps)^n$ their Lee distance is defined to be 
\begin{align*}
    \LD(x, y) = \LW(x-y).
\end{align*}
Note that the Hamming weight of an element $x \in (\zps)^n$ is a natural lower bound on its Lee weight. Additionally, the Lee weight of each entry $x_i$ can never exceed $M := \floor{p^s/2}$. Hence, we have
\begin{align*}
    0 \leq \HW(x) \leq \LW(x) \leq \HW(x)M \leq n M.
\end{align*}


Considering the Lee metric for linear codes $\code$ over $\zps$, we are able to introduce the minimum Lee distance of $\code$.
\begin{definition}\label{def:minDistL}
    Given a linear code $\code \subseteq (\zps)^n$, we define its \textit{minimum Lee distance} as
    \begin{align*}
        \LD(\code) := \min\set{ \LW(c) \st c \in \code\setminus\set{0}}.
    \end{align*}
\end{definition}
Again, we easily observe that 
$$\HD(\code) \leq \LD(\code) \leq M \HD(\code).$$

The linear isometries of the Lee metric are $\{-1,1\} \rtimes S_n,$ that is we can permute the entries and multiply each entry with either 1 or $-1.$

\subsection{Singleton-like Bounds over $\zps$}\hfill\\
One of the major research directions in coding theory is to bound the minimum distance of a code. 
This will then, in turn, bound the error-correction capability of a code $\code \subseteq \zpsn$. In fact, the higher the minimum distance of $\code$, the more errors can be corrected.

The task of bounding the minimum distance, clearly depends on the metric  used to endow the ambient space. 
For the Hamming metric over finite fields, the best-known bound is the Singleton bound describing the trade-off between the minimum Hamming distance of a code and its dimension.
\begin{theorem}[Singleton Bound]
    Given a linear code $\code \subseteq \field_q^n$ of dimension $k$ over $\mathbb{F}_q$, its minimum Hamming distance is upper bounded by
    \begin{align*}
        \HD(\code) \leq n - k + 1.
    \end{align*}
\end{theorem}
Codes achieving this bound are called \textit{maximum distance separable} (MDS). It is well known that for $q$ tending to infinity a random linear code will attain the Singleton bound. Hence, MDS codes are dense for $q$ tending to infinity with high probability. On the other hand, if we let the length $n$ grow, MDS codes are sparse. 
 Note that this also follows immediately from the famous MDS conjecture \cite{segre}, which states that if $q$ is odd, then an MDS code must have $n \leq q+1$. In the case where $q=2^s$, and $k = 3$ or $k = q -1$, in which case $n \leq q + 2$. 
A very analogous bound can be established over finite integer residue rings $\zps$ for the Hamming metric. 
\begin{proposition}\label{MDS}
     Let $\code \subseteq \left( \mathbb{Z}/p^s\mathbb{Z}\right)^n$ be a linear code of $\mathbb{Z}/p^s\mathbb{Z}$-dimension $k$, then 
     $$\HD(\code) \leq n-k+1.$$
\end{proposition}
This Singleton-like bound does hold for non-linear codes as well, stating that 
$\mid \mathcal{C} \mid \leq q^{n- \HD(\code) +1}$. For linear codes only, this bound has been further tightened:
\begin{proposition}[\cite{doughertybook, douandshi})]\label{MDR}
     Let $\code \subseteq \left( \mathbb{Z}/p^s\mathbb{Z}\right)^n$ be a linear code of rank $K$, then 
     $$\HD(\code) \leq n-K+1.$$
\end{proposition}

To differentiate codes achieving the bound in Proposition \ref{MDR} from MDS codes, we will refer to the first as \textit{maximum distance codes with respect to the rank} (or MDR codes for short). However, MDS codes and MDR codes are closely related. In fact, any linear MDS code is always an MDR code. Vice versa we can say that a code $\code \subseteq \zpsn$ is MDR if and only if the socle $\code \cap \langle p^{s-1} \rangle$ can be identified with an MDS over $\field_p$. We can therefore directly apply the results from finite fields which means that by the MDS conjecture, MDR codes are sparse as $n$ grows large and, as usual, they are dense as $p$ grows large.

Note that the relation $$\HD(\code) \leq \LD(\code) \leq M \HD(\code)$$
immediately gives raise to a bound for  codes in the Lee metric, i.e., $$\LD(\code) \leq M(n-\lfloor k \rfloor +1).$$

Apart from this obvious bound,  the first Lee-metric Singleton-like bound was introduced by Shiromoto in 2000 \cite{shiromoto}. 
\begin{theorem}[\cite{shiromoto}]\label{shir}
    Consider a linear code $\code \subseteq \zpsn$ of $\mathbb{Z}/p^s\mathbb{Z}$-dimension $k$. Then the following bound holds
    \begin{align*}
        \floor{ \frac{\LD(\code) - 1}{M} } \leq n - k.
    \end{align*}
\end{theorem}
Note that this bound implies $\LD(\code) \leq M(n-k)+\alpha$ for some $\alpha \in \{1, \ldots, M\}$, and can be generalized easily to  $\LD(\code) \leq M(n-\lfloor k \rfloor)+\alpha$, again for some $\alpha \in \{1, \ldots, M\}$. 

The bound from Theorem \ref{shir} can be derived from the simple fact that $\LD(\code) \leq M \HD(\code)$ for a code $\code \subseteq \zpsn$ or by a puncturing argument, as in the classical case. 
That is, given a code of size $\mid \code \mid$ and minimum Lee distance $\LD(\code)$, we can  puncture the code in only $ \floor{ \frac{\LD(\code) - 1}{M} }$ coordinates, to make sure that the punctured code $\mathcal{C}'$ still has size $\mid \mathcal{C} \mid.$ In fact, every pair of codewords in $\mathcal{C}$ has Lee distance at least $\LD(\code)$, by puncturing in one position, one has to assume that their distance decreased by the maximal possible value, i.e., by $M$.
The claim then follows easily as
$$\mathcal{C}' \subseteq \left(\mathbb{Z}/p^s\mathbb{Z}\right)^{n- \floor{ \frac{\LD(\code) - 1}{M} }}.$$

Shiromoto also provided a code which attains the bound, showing its tightness. 
\begin{example}
    Let $\mathcal{C}= \langle 1,2 \rangle \subseteq \left(\mathbb{Z}/5\mathbb{Z}\right)^2$. Then, since $\LD(\code)=3$, this code attends the bound of Theorem\ref{shir}, as
       $$ \floor{ \frac{\LD(\code) - 1}{M} } = \frac{3-1}{2}=1=  n - k =2-1.$$
\end{example}
However, in \cite{byrne2023bounds}, it was observed that this is actually the only non-trivial linear code that attends the bound in Theorem \ref{shir}. The Lee-metric Singleton bound can be further improved, as shown in \cite{byrne2023bounds}, e.g. by employing the rank $K$.

\begin{corollary}\label{LMDR}
   Let $\code \subseteq \left(\mathbb{Z}/p^s\mathbb{Z}\right)^n$ be a linear code of rank $K$, then 
   $$\LD(\code) \leq M(n-K+1).$$
\end{corollary} 

Alderson and Huntemann in \cite{alderson} provided a similar bound to Corollary \ref{LMDR} by restricting $k$ to a positive integer bounded by $n$.
\begin{theorem}[\cite{alderson}]\label{alderson}
For any code $\code$ in $\zpsk{n}$ of  $\mathbb{Z}/p^s\mathbb{Z}$-dimension $k$ a positive integer, $1<k<n$ we have that $$\LD(\code) \leq M(n-k).$$
\end{theorem}

However, in \cite{byrne2023bounds}, the authors characterized all codes attaining the above Lee-metric Singleton bounds, with the result that their optimal codes are sparse in both cases, i.e., when $n, p$ or $s \to \infty$.

 Using the newly introduced parameter for the code, i.e., the support subtype, one can easily derive an improved Lee-metric Singleton bound from the puncturing argument. 
     
     \begin{theorem}\label{thm:SBpunct}
     Let $\code \subseteq \left(\mathbb{Z}/p^s\mathbb{Z}\right)^n$ be a linear code of rank $K$ and support subtype $(n_0, \ldots, n_{s-1},0).$ Define for all $i \in \{0, \ldots, s\}$ \begin{align*}
         M_i = \floor{\frac{p^{s-i}}{2}}p^i, \ \
         B_j  = \sum_{i=j}^{s-1} n_{i}, \ \
         A_j  = \sum_{i=j}^{s-1} n_{i}M_{i}.
     \end{align*}
     Let $j \in \{1, \ldots, s-1\}$ be the smallest positive integer such that $A_j < \LD(\code)$, then
     $$K \leq n-B_j- \left\lfloor \frac{\LD(\code)-A_j-1}{M_{j-1}}\right\rfloor.$$
     \end{theorem}

     \begin{proof}
        We start by puncturing the code in the positions of smallest possible Lee weight. To identify these positions, we use the support subtype. Clearly, in the ideal $\langle p^{i} \rangle$, we have as largest possible Lee weight $M_i= \lfloor \frac{p^{s-1}}{2} \rfloor p^i$, and thus we would start puncturing in the positions, where all codewords live in $\langle p^{s-1} \rangle$, i.e., in the positions belonging to the support subtype $n_{s-1}.$
        We hence assume that the minimum distance between two distinct tuples decreased by $A_{s-1}=n_{s-1}M_{s-1}$. If this is still smaller than the minimum Lee distance, we can continue puncturing in the next ideal, namely $\langle p^{s-2} \rangle.$
        We continue in this fashion, every time puncturing in $n_iM_i$ positions, until $A_j=\sum_{i=j}^{s-1} n_{i}M_{i}$ has reached the minimum Lee distance. We are left with codewords that are at least $\LD(\code)-A_j$ apart, thus we can continue puncturing in $\left\lfloor \frac{\LD(\code)-A_j-1}{M_{j-1}}\right\rfloor$ positions living in $\langle p^{j-1}\rangle$, i.e., belonging to the support subtype $n_{j-1}$, and still be sure that the punctured code has the same size as the original code. 
         In this case, we have the new length of the punctured code, being $n-B_j -\left\lfloor \frac{\LD(\code)-A_j-1}{M_{j-1}}\right\rfloor$, for $B_j= \sum_{i=j}^{s-1} n_{i}.$  
     \end{proof}
     \begin{example}
        Let  us consider $\mathcal{C} \subseteq \left(\mathbb{Z}/9\mathbb{Z}\right)^4$  generated by 
        $$G = \begin{pmatrix}
            1 & 0& 2 & 3 \\  0 & 3 & 6 & 0 \\ 0 & 0 & 3 & 6 
        \end{pmatrix}.$$
        The $\mathbb{Z}/9\mathbb{Z}$-dimension of this code is $k=2$.
        The minimum Lee distance of this code is $\LD(\code)=6$. The code also has support subtype $(2,2,0).$ Thus, we would identify $j=s=2$, as we cannot puncture in both positions belonging to $n_1=2$, namely the second and the last column, as we would get $n_1 M_1=6 \not< \LD(\code).$ However, we can puncture in one of these two columns. In fact, $\left\lfloor \frac{\LD(\code)-0-1}{3}\right\rfloor = 1$. That is, the bound in Theorem \ref{thm:SBpunct} is attained as $$K=3 =4-0-\left\lfloor \frac{6-0-1}{3}\right\rfloor = n-B_j- \left\lfloor \frac{\LD(\code)-A_j-1}{M_{j-1}}\right\rfloor.$$
        The bound from Theorem \ref{shir} would instead give
        $$\left\lfloor \frac{\LD(\code)-1}{M} \right\rfloor = \left\lfloor \frac{6-1}{4}\right\rfloor =1 < 2 = n- k.$$
     \end{example}
 Since we are also in the case where $k$ is an integer strictly larger than 1, we can also apply the bound from Theorem \ref{alderson}, and get
 $$\LD(\code) =6 < 8= (4-2) \cdot 4  = (n-k)M.$$

 We can rewrite the bound from Theorem \ref{thm:SBpunct} as upper bound on the minimum Lee distance as, for $j$ the smallest positive integer with $A_j < \LD(\code)$ we have
 $$\LD(\code) \leq M_{j-1} \left(\sum_{i=0}^{j-1} n_i -K\right) +\sum_{i=j}^{s-1} n_i M_i + \alpha$$ for some $\alpha \in \{1, \ldots, M_{j-1}\}.$
However, the condition to find the smallest $j$ such that $A_j < \LD(\code)$, renders the bound impractical, as usually one does not know the minimum Lee distance of the code and thus wants to bound it.
    
\section{Generalized Hamming Weights 
}\label{sec:genWeight_subcodes}
In this section we propose a new definition of generalized Lee weights. We base the definition on the Hamming weight counterparts. Generalized Hamming weights have originally been introduced in \cite{gen} and were then rediscovered by Wei in \cite{wei1991generalized}. Generalized weights in the Hamming metric have been studied in various areas \cite{cardell2020generalized, dougherty2002generalized, gorla2022generalizedcolumn, gorla2023generalized, ravagnani2016generalized}. In \cite{gorla2022generalized} the authors defined and the generalized Hamming weights of ring-linear codes by considering the join-Hamming support of a code.
        
    Let us recap the definition of a Hamming support of a vector and a code. For this, consider a finite field $\field_q$ of $q$ elements and a positive integer $n$. The \textit{Hamming support} of $x \in \mathbb{F}_q^n$ is defined to be the set of indices where $x$ is nonzero, i.e.,
    \begin{align*}
        \suppH (x) := \set{ i \in \set{1, \ldots, n} \st x_i \neq 0 }.
    \end{align*}
    Note here, that the cardinality of the support of $x$ corresponds to the Hamming weight of $x$, that is 
    \begin{align}\label{Hamming_weight_vec}
        \card{ \suppH(x) } = \HW(x).
    \end{align}
    Let us consider now a linear code $\code \subset \Fq^n$ of length $n$ and dimension $K$. For the definition of the Hamming support of $\code$ we take into account every codeword $c\in \code$ and  for each index, we figure out whether a codeword exists which is nonzero in this position, i.e.,
    \begin{align}\label{Hamming_support_code}
        \suppH(\code)= \{i \in \{1, \ldots, n \} \mid \exists c \in \code, c_i \neq 0\}.
    \end{align}
    Analogously to the definition of the weight of a vector $x$ in \eqref{Hamming_weight_vec}, we can define the weight of a code $\code$ to be the size of its support, i.e.,
    \begin{align}\label{Hamming_weight_code}
        \wt(\code)= \card{ \supp(\code) }.
    \end{align}

    The goal is to generalize these notions to other metrics and ambient spaces. In particular, we are interested in the Lee metric defined over rings. In order to do so, we will follow the approach of \cite{gorla2022generalized}.
    Let $\mathcal{R}$ denote a finite unitary ring.
    Let us consider a weight function 
    $$\wt: \mathcal{R} \to \mathbb{N},$$ which  is such that 
    \begin{enumerate}
        \item $\wt(0)=0$ and $\wt(x)>0$ for all $x \neq 0$, \label{item1}
        \item $\wt(x)=\wt(-x)$,
        \item $\wt(x+y) \leq \wt(x)+\wt(y).$ \label{item3}
    \end{enumerate}
    Note the difference to the definition used in \cite[Definition 2.4]{gorla2022generalized} which corresponds to a  norm, as they also require the absolute homogeneity property. That is, for any $\lambda \in \mathcal{R}\setminus\{0\}$, we have that 
    \begin{itemize} 
        \item[4.] $\wt(\lambda x)= \wt(x).$
    \end{itemize}
    
    This property holds for the Hamming weight and many other weights, however it is not a requirement for a weight function. The properties \ref{item1}.-\ref{item3}. are enough to induce a distance. In fact, for a weight function $\wt$ with properties \ref{item1}.-\ref{item3}. we can define a distance as 
    \begin{align*}
        \dist: \mathcal{R}\times \mathcal{R} &\to \mathbb{N}, \\
        (x,y) & \mapsto \wt(x-y).
    \end{align*}

    By abuse of notation, we will also denote their coordinate-wise extension by $\wt$ and $d$, that is for $x,y \in \mathcal{R}^n$ we have
    $$\wt(x)= \sum_{i=1}^n \wt(x_i) \text{ and } \dist(x,y)= \sum_{i=1}^n \dist(x_i,y_i).$$ We call such weight functions \textit{additive weights}.
    Given a weight function, one can then define the \textit{support} of $x \in \mathcal{R}^n$ as an $n$-tuple
    \begin{align*} 
        \supp (x) := \left( \wt(x_1), \ldots, \wt(x_n) \right).
    \end{align*}
    
    As we are now dealing with an $n$-tuple instead of a subset of $\{1, \ldots, n\}$, we require some additional definitions. 
    Let $s ,t\in \mathbb{N}^n$ be  $n$-tuples. 
    The size of $s$ is given by the sum of its entries, that is $\card{s} = \sum_{i=1}^n s_i.$
    The \textit{join} of $s$ and $t$, denoted by $s \vee t$ is given by taking the maximum in each position, that is 
    $$s \vee t= ( \max\{s_1,t_1\}, \ldots, \max\{s_n,t_n\}).$$
    The \textit{meet} of $s$ and $t$, denoted by $s \wedge t$ is given by taking the minimum in each position, that is
    $$s \wedge t = (\min\{s_1, t_1\}, \ldots, \min\{s_n,t_n\}).$$
    Since the weight is additive, we have that 
    $$\card{ \supp(x) } =\sum_{i=1}^n \wt(x_i)= \wt(x).$$
    The Hamming support is thus such an example, where instead of considering the support as subset of $\{1, \ldots, n\}$, the support is considered as the $n$-tuple
    \begin{align*}
    \suppH (x) = \left( \HW(x_1), \ldots, \HW(x_n) \right).
    \end{align*}
    
    In order to extend this to the support of codes, we have several options. One of those, is the join-support, as considered in \cite{gorla2022generalized}:
    for $\code \subseteq \mathcal{R}^n$ a linear code, we define its \textit{join-support} as
     \begin{align*}
        \suppjoin (\code) := \left( \max_{c \in \code} \wt(c_1), \ldots , \max_{c \in \code} \wt(c_n) \right)= \bigvee\limits_{c \in \code} \supp(c).
    \end{align*}
    Note that another possibility would be to define the \textit{meet-support}, as follows
    \begin{align*}
        \suppmeet(\code) :=& \left( \min_{c \in \code} \{\max\{\wt(c_1),0\}\}, \ldots, \min_{c \in \code} \{\max\{\wt(c_n),0\}\}\right)\\  =& \bigwedge_{c \in \code} (\supp(c) \vee 0).
    \end{align*}
    As the Hamming weight of nonzero elements equals one, we observe that the join-support coincides with the meet-support of a code $\code$ in the Hamming metric, i.e., 
    \begin{align*}
        \suppHjoin(\code) = \suppHmeet(\code).
    \end{align*}
    
    \begin{example}
        Let us consider the code over $\mathbb{F}_3$ generated by $$G= \begin{pmatrix} 1 & 0 & 0 & 1 & 0 \\ 0 & 1 & 0 & 1 & 0 \\ 0 & 0 & 1 & 0 &0 \end{pmatrix}.$$ 
        With the usual definition of the Hamming support in \eqref{Hamming_support_code}, we have that 
        $$\suppH(\code)= \{1,2,3,4\}. $$
        With the join-support, we are considering the maximal value of the weight of the entries of a codeword in each position, that is 
        $$\suppHjoin(\code)=(1,1,1,1,0).$$
        For the meet-support, we take the minimum nonzero value of the weight of the entries of a codeword in each position, which also gives $(1,1,1,1,0).$\\
        By applying definition of the weight of a code \eqref{Hamming_weight_code} we observe that all the three support definitions of $\code$ yield the same weight 
        \begin{align*}
            \HW(\code) &= \card{\suppH(\code)} = 4, \\
            \HWjoin(\code) &= \HWmeet(\code) = 4.
        \end{align*}
    \end{example}

In the classical case, one defines the $r$-th generalized weights as follows.
\begin{definition}
    Let $\code \subseteq \mathbb{F}_q^n$ be a linear code of  dimension $k$. Then for any $r \in \{1, \ldots, k\}$ the $r$-th generalized weight is given by
    $$\dist^r(\code)= \min\{\text{wt}(\mathcal{D}) \mid \mathcal{D} \leq\code, \text{dim}(\mathcal{D}) = r\}.$$
\end{definition}
    
For the generalized weights, we want the following properties to hold. 
    Let $\code$ be a linear code of dimension $k$. Then we have
    \begin{enumerate}
        \item $\dist(\code) = \dist^1(\code),$
        \item $\dist^r(\code) < \dist^{r+1}(\code)$ for every $1 \leq r < k$,
        \item $\dist^k(\code) = \wt(\code)$.
    \end{enumerate}
    For the Hamming support and the rank support these properties have been showed in \cite{wei1991generalized, RelativeRankweight, rankGenWeight}. In \cite{gorla2022generalized} they get similar properties, with the exception of $\dist^r(\code) \leq \dist^{r+1}(\code),$ instead of the strict inequality.
    
   The strict inequality is important to us, however, as it then leads to neat Singleton bounds, that is as
    $$\dist(\code)=\dist^1(\code) < \dist^2(\code) < \cdots < \dist^{k-1}(\code) <\dist^k(\code)  = \wt(\code), $$
we get $$\dist(\code) \leq \wt(\code)-k+1. $$
Note that for non-degenerate codes we have that $\HW(\code)=n$, and thus we retrieve the classical Singleton bound $$\HD(\code)\leq n-k+1.$$

Since we move from the classical case of finite fields to rings, we have to exchange the fixed dimension of the subcodes with a ring-analogue parameter. A natural choice would be the $\mathbb{Z}/p^s\mathbb{Z}$-dimension, but as this value is not necessarily an integer and there might not exist subcodes of $\mathcal{C}$ of certain fixed smaller rational number as the $\mathbb{Z}/p^s\mathbb{Z}$-dimension, we choose to discard this option. 

In \cite{dougherty2002generalized}, the authors chose to exchange the dimension with the subtype. 
In fact, in the same paper the authors defined generalized Lee weights for $\mathbb{Z}/4\mathbb{Z}.$ This particular case is, however, not of interest for us, as the Lee-metric Singleton bound over $\mathbb{Z}/4\mathbb{Z}$ directly follows from the Gray isometry.

Following the idea of \cite{dougherty2002generalized},  a first attempt on defining generalized weights over $\mathbb{Z}/p^s\mathbb{Z}$ would be the following. 
\begin{definition}
        Let $\code \subseteq \left(\mathbb{Z}/p^s\mathbb{Z}\right)^n$ be a linear code of subtype $(k_0, \ldots, k_{s-1})$. Then for any $(r_0, \ldots, r_{s-1})$ with $r_i \leq k_i$ for all $i \in \{1, \ldots, s-1\}$ the  $(r_0, \ldots, r_{s-1})$-th generalized weight is given by
    $$\dist^{(r_0, \ldots, r_{s-1})}(\code)= \min\{\text{wt}(\mathcal{D}) \mid \mathcal{D} \leq\code,  \mathcal{D} \text{ has subtype } (r_0, \ldots, r_{s-1})\}.$$
\end{definition}



Note that this definition is not considering all possible subcodes or all possible subtypes of subcodes. 

To allow for a comparison between two different subtypes $(r_0, \ldots, r_{s-1})$ and $(r_0', \ldots, r_{s-1}')$ which might have $r_i<r_i'$ for some $i$ but $r_j > r_j'$ for some $j$, a natural choice is 
to impose a lexicographical order, i.e., we consider the order $$(k_0,\ldots, k_{s-1}) >(k_0-1,\ldots, k_{s-1}) > \cdots >  (0,k_1, \ldots, k_{s-1}) >    \cdots > (0,\ldots, 0, 1).$$ However, then the property $\dist(\code)=\dist^{(0, \ldots, 0,1)}(\code)$ is not guaranteed. 

In fact, a minimum Lee weight codeword will live in a subcode having subtype one of the standard vectors $e_i$.  Thus, we have $\dist(\code)=   \dist^{e_i}(\code)$ for some $i$. 
%
Observing that this just means to fix the rank of the subcode as 1,  we choose to directly fix the rank instead. 

\begin{definition}
        Let $\code \subseteq \left(\mathbb{Z}/p^s\mathbb{Z}\right)^n$ be a linear code of of rank $K$. Then for any $r \in \{1, \ldots, K\}$ the  $r$-th generalized weight is given by
    $$\dist^{(r}(\code)= \min\{\text{wt}(\mathcal{D}) \mid \mathcal{D} \leq\code,  \text{rk}(\mathcal{D})=r \}.$$
\end{definition}

\section{Join-Support in the Lee Metric}\label{sec:joinsupport}
Now we turn our focus on exchanging the Hamming weight with the Lee weight.
We want to define the Lee support and hence the generalized Lee weights in a similar fashion. 
For the Lee support of a vector $x \in (\zps)^n$ we view the Lee support as an $n$-tuple and define it analogous to the Hamming support, i.e.,
\begin{align*}
    \suppL (x) := (\lweight(x_1), \ldots ,\lweight(x_n)).
\end{align*}
We now want to define the Lee support and the generalized Lee weights of a code $\code \subseteq \mathbb{Z}/p^s\mathbb{Z}^n$, according to the two possibilities: the join-support and the meet-support. However, in the Lee metric, the meet-support is not practical to derive bounds on the minimum distance for the code. Let us quickly argue why.
 \begin{definition}\label{def:meetLeeSupp}
    For a code $\code \subset (\zps)^n$ we define the \textit{Lee meet-support} as the minimal (if possible) nonzero Lee weight in each position among all codewords, meaning that
    \begin{align*}
        \suppLmeet (\code) := \left(\min_{c \in \code} \{ \max \{ \lweight(c_1), 0\}\} , \ldots , \min_{c \in \code} \{ \max \{ \lweight(c_n), 0 \}\} \right).
    \end{align*}
\end{definition}
   \begin{proposition}
    For $\code \subseteq (\zps)^n$ of support subtype $(n_0, \ldots, n_s)$, we have that
    \begin{align*}
        \card{ \suppLmeet(\code)} &= \LW(\code)= \sum_{i=0}^s n_ip^i.
    \end{align*}
\end{proposition}
    
\begin{proof}
    The Lee meet-support asks to take the smallest nonzero Lee weight in position $j$ and then to sum over all entries $j \in \{1, \ldots, n\}.$ Since any position belonging to the support subtype $n_i$ is living in the ideal $\langle p^i \rangle$, this position has as smallest nonzero Lee weight $p^i$.
\end{proof}
    We can then define the $r$-th generalized meet-Lee weights.
    
        \begin{definition}
        Let $\code \subseteq \left(\mathbb{Z}/p^s\mathbb{Z}\right)^n$ be a linear code of rank $K$.
        For $i \in \{1, \ldots, K\}$ define the $i$th generalized  meet-Lee weight as
        $$\LDmeet^i(\code)= \min \{ \card{\suppLmeet(\subcode)} \st \subcode \leq \code,  \rk(\mathcal{D})= i\}.$$
    \end{definition}
    
Unfortunately, the meet-support in the Lee metric does not generally fulfill the property    $$\LD(\code) \leq \LDmeet^1(\code).$$
     
As an easy example for $\LD(\code) > \LDmeet^1(\code)$, consider $\code=\langle (1,2) \rangle \subseteq \mathbb{Z}/9\mathbb{Z}^n$. The minimum Lee distance of this code is 3, however, the first generalized meet-Lee weight is 2, as $\suppLmeet \langle(1,2)\rangle =(1,1)$ is the minimal meet-Lee support. This will then not lead to a Singleton bound and is thus discarded.\\
     
Instead we will now focus on the join-Lee support, as also promoted in \cite{gorla2022generalized}.
 
   \begin{definition}\label{def:joinLeeSupp}
    For a code $\code \subset (\zps)^n$ its \textit{join-Lee support} is defined as the maximal possible Lee weight in each position among all codewords, i.e.,
    \begin{align*}
        \suppLjoin (\code) := \left(\max_{c \in \code} \lweight(c_1), \ldots , \max_{c \in \code} \lweight(c_n)\right).
    \end{align*}
\end{definition}

\begin{proposition}
    For $\code \subseteq (\zps)^n$ of support subtype $(n_0, \ldots, n_s)$, we have that
    \begin{align*}
        \card{ \suppLjoin(\code)} &= \LWjoin(\code)= \sum_{i=0}^s n_iM_i.
    \end{align*}
\end{proposition}
\begin{proof}
    In each index $j \in \{1, \ldots, n\}$, we can check in which ideal this coordinate of the code lives. Let us assume that this is $\langle p^i \rangle$, for some $i \in \{0, \ldots, s\}$. Since the support of the code takes the maximum over all codewords in the code, we will reach in this entry the maximal Lee weight of the ideal $\langle p^i \rangle,$ which is given by $M_i= \lfloor \frac{p^{s-i}}{2}\rfloor p^i$. Since we know the support subtype of the code, we have $n_i$ many of this entries.
\end{proof}

The $r$-th generalized join-Lee weight is then defined as follows.

  \begin{definition}
        Let $\code \subseteq \left(\mathbb{Z}/p^s\mathbb{Z}\right)^n$ be a linear code of rank $K$.
        For $i \in \{1, \ldots, K\}$ we define the $i$th generalized  join-Lee weight as
        $$\LDjoin^r(\code)= \min\{ \LWjoin(\mathcal{D}) \mid \mathcal{D} \leq \code, \text{rk}(\mathcal{D})=r\}.$$
    \end{definition}

Let us consider an example, which also perfectly shows the differences between the meet-Lee support and the join-Lee support. 

\begin{example}
Let us consider the code $\code \subseteq \mathbb{Z}/9\mathbb{Z}^4$ generated by 
$$G= \begin{pmatrix} 1 & 0 & 3 & 2 \\ 0 & 1 & 2 & 0 \\ 0 & 0 & 3 & 3 \end{pmatrix},$$ which has support subtype $(4,0,0)$ and minimum Lee distance 2, for example $(1,0,0,8)$ is a minimal Lee weight codeword. 
For the generalized meet-Lee weights we have that
$$\LD(\code) \geq \LDmeet^1(\code)  \leq \LDmeet^2(\code)=\LDmeet^3(\code) =  \LWmeet(\code).$$
Since \begin{align*}  \LDmeet^1(\code) & =\LWmeet(\langle (0,1,2,0)\rangle= 2  \\
\LDmeet^2(\code)& = \LWmeet(\langle (1,0,3,2),(0,1,2,0)\rangle)= 4 \\ 
\LDmeet^3(\code) & = \LWmeet(\langle G \rangle)= 4 = \LWmeet(\code).
\end{align*}

For the generalized join-Lee weights we have that 
$$\LD(\code) \leq \LDjoin^1(\code)  < \LDjoin^2(\code)  < \LDjoin^3(\code) \leq \LWjoin(\code).$$

Since \begin{align*}
  \LDjoin^1(\code)& =\LWjoin(\langle (0,0,3,3) \rangle)= 6 \\
  \LDjoin^2(\code) & = \LWjoin(\langle (0,0,3,3),(3,0,0,6) \rangle )= 9 \\
  \LDjoin^3(\code) & = \LWjoin(\code \cap \langle 3\rangle)=12 \\
   \LWjoin(\code) &= 16.
    \end{align*}
\end{example}

    \begin{proposition}
        The subcodes which attain the $r$-th generalized join-Lee weights all live in the socle.
    \end{proposition}
    \begin{proof}
        By contradiction, assume that $\subcode \leq \code$ of rank $r$ achieves the $r$-th generalized Lee weight $\LDjoin^r(\code)$ and $\mathcal{D}$ does not live in the socle. That is, if  $\subcode$ has support subtype $(n_0, \ldots, n_s)$, then for some $i < s-1$ we have $n_i \neq 0.$ Thus, 
        \begin{align*}
            \LDjoin^r(\code)= \card{ \suppLjoin(\subcode)} \leq \sum_{i=1}^{s-1} n_iM_i.
        \end{align*}
        By considering the subcode $\subcode_0 = \subcode \cap \langle p^{s-1} \rangle$, we observe that its support subtype is $(0, \ldots, 0, n_0 + \cdots + n_{s-1}, n_s).$ Furthermore,
        \begin{align*}
            \LWjoin(\subcode_0) = M_{s-1}(n_0 + \cdots +n_{s-1}) < \sum_{i=1}^{s-1} n_iM_i,
        \end{align*}
        since $M_{s-1}<M_i$ for all $i < s-1.$ This gives a contradiction to the minimality of the subcode $\subcode$.
    \end{proof}

Thus, it is enough to only consider the generalized join-Lee weights of the socle $\code \cap \langle p^{s-1} \rangle.$

\begin{corollary}\label{cor:socle}
    Let $\code \subseteq \mathbb{Z}/p^s\mathbb{Z}^n$ be a linear code of rank $K.$ Then for all $r \in \{1, \ldots, K\}$ we have 
    $$\LDjoin^r(\code) = \LDjoin^r(\code \cap \langle p^{s-1} \rangle).$$
\end{corollary}
This property gives us an immediate relation to the generalized Hamming weights.
In fact, the socle can be considered as a code over $\mathbb{F}_p$ and the subcodes which attain the minimal join-Lee support are then those which attain the minimal Hamming support.

\begin{corollary}\label{cor:dLtodH}
    Let $\code \subseteq \mathbb{Z}/p^s\mathbb{Z}^n$ be a linear code of rank $K.$ Then for all $r \in \{1, \ldots, K\}$ we have 
    $$\LDjoin^r(\code) = \HD^r(\code) M_{s-1}.$$ 
\end{corollary}
Thus, we can use the properties of the generalized Hamming weights to show the following. 

    \begin{proposition}\label{cor:dLprop}
        Let $\code \subset (\zps)^n$ be a linear code of rank $K$. Then we have
        \begin{enumerate}
            \item $\LD(\code) \leq \LDjoin^1(\code)$.
            \item $\LDjoin^r(\code) < \LDjoin^{r+1}(\code)$ for every $1 \leq r < K$.
            \item $\LDjoin^K(\code) \leq \LWjoin(\code)$.
        \end{enumerate}
    \end{proposition}
    \begin{proof}
        The first property follows easily from the definition of the join- Lee support of a vector $x$. It can be tight, whenever the minimal Lee weight codeword is in the socle, which is not necessary. 
        For the second property we simply use Corollary \ref{cor:dLtodH} and the third property also simply follows from the definition of join-Lee support.
%
    \end{proof}
In fact, we do not recover the exact properties of the generalized Hamming weight codes. We do not have $\LD(\code)= \LDjoin^1(\code)$ and $\LWjoin(\code)= \LDjoin^K(\code).$ This seems to be  the price we have to pay in order to drop the absolute homogeneity property and to be able to consider the Lee metric. 
However, unlike the meet-Lee support we get a nice chain of inequalities:
$$\LD(\code) \leq \LDjoin^1(\code) < \LDjoin^2(\code) < \cdots < \LDjoin^K(\code) \leq \LWjoin(\code).$$ 
This gives us a new Lee-metric Singleton bound.
 
\begin{theorem}\label{thm:SB_joinsupp}
    Let $\code \subset (\zps)^n$ be a (non-degenerate) linear code of rank $K$. Then we have
    $$\LD(\code) \leq M_{s-1}(n-K+1)= \floor{ \frac{p}{2}} p^{s-1} (n-K+1).$$
\end{theorem}
\begin{proof}
    Using the properties 1.-3. from Proposition \ref{cor:dLprop} we know that
    $$\LD(\code) \leq \LDjoin^K(\code) -\sum_{i=1}^{K-1} x_i,$$ 
    where $$x_i =  \LDjoin^i(\code) - \LDjoin^{i-1}(\code).$$
    From Corollary \ref{cor:dLtodH} we know that 
    $$x_i=\LDjoin^i(\code) - \LDjoin^{i-1}(\code) \geq M_{s-1}.$$ We get the claim using that  $$\LDjoin^K(\code)= \sum_{i=0}^{s-1} n_i M_{s-1} = nM_{s-1}, $$ where we have assumed that the code is non-degenerate.
    Note that  
    we could have gotten this bound also by directly using 
    $$\LD(\code) \leq\LDjoin^1(\code) = \HD^1(\code) M_{s-1} = \HD(\code) M_{s-1} \leq (n-K+1) M_{s-1}.$$ 
\end{proof}
 
 This new Singleton bound is clearly much sharper than the previously known Lee-metric Singleton bounds, for example the bound from Theorem \ref{shir}. 

\subsection{Density of Optimal Codes with respect to the Join-Lee Support}\hfill\\
Clearly, any code $\code \in (\zps)^n$ of rank $K$ attaining this bound can be characterized by the following two properties:
\begin{enumerate}
    \item The socle $\code_{s-1}=\code \cap \langle p^{s-1} \rangle$ is an MDS code over $\mathbb{F}_p. $
    \item There exists a  minimum Lee weight codeword in the socle. 
\end{enumerate}
The first property already implies sparsity as $n$ tends to infinity and triviality for $p=2.$ Even the second property is problematic: $\LD(\code_{s-1})=(n-K+1)M_{s-1}$, implies that all nonzero entries of a minimal Hamming weight codeword in the socle must be of maximal Lee weight. Using the systematic form of the socle,
\begin{align*}
    G_{s-1} = \begin{pmatrix} p^{s-1} \mathbb{I}_K & p^{s-1}A\end{pmatrix},
\end{align*}
we can immediately see that any row $g$ of $G_{s-1}$ is also of minimal Hamming weight $n-K+1$. Thus, for $g$ to have only nonzero entries of maximal Lee weight implies $p^{s-1}=M_{s-1}$, which will restrict optimal codes with respect to this bound to $p\in \set{2,3}$ and any positive integer $s$. Because of the MDS property over $\field_3$, we must have a block length $n\leq 4$. \\
We can drop the second condition, i.e., there exists a minimal Lee weight codeword in the socle, if we manage to estimate the difference 
$$\LDjoin^1(\code) - \LD(\code).$$
This task is, however, equally hard as bounding $\LD(\code)$ itself.\\

Knowing hence, that only codes over $\zps$ for $p = 2, 3$ with length $n = 4$ can attain this bound, 
the socles of the codes $\code \subseteq \zpsn$ with $p = 3$ and $n \leq 4$ attaining the bound in Theorem \ref{thm:SB_joinsupp} are hence generated by a matrix of the form 
\begin{align}\label{equ:MLD_join}
    \begin{pmatrix}
        3^{s-1} & 0 & 0 & A \\ 0 & 3^{s-1} & 0 & B \\ 0 & 0 & 3^{s-1} & C
    \end{pmatrix},
\end{align}
where $A, B, C \in \zps$ are such that their Lee weight is $M_{s-1}$.\\

\begin{example}
For instance, the code $\code\subseteq(\Zp{9})^4$ of rank $K = 3$ generated by $$\begin{pmatrix} 3 & 0 & 0 & 3 \\ 0 & 3 & 0 & 6 \\ 0 & 0 & 3 & 6 \end{pmatrix}.$$
This code has $\LD(\code)=6$ and thus attains the join-Lee metric Singleton bound as $\LD(\code)=6= 3 \cdot (4-3+1)= M_{s-1}(n-K+1).$
\end{example}

Note that for MDS codes, we actually know all $r$-th generalized Hamming weights: let $\mathcal{C} \subseteq \mathbb{F}_q^n$ be a linear code of dimension $k$, then
$$\HD^r(\code)=n-k+r.$$
Thus, a natural question that arises is whether the optimal codes with respect to the newly defined Lee-metric Singleton bound have a similar behaviour. That is, we  are interested  in an expression for the $r$-th generalized join-Lee weight $\LDjoin^r(\code)$ for every $r \in \{1, \ldots , K\}$.

\begin{proposition}
    Let $\code \subseteq (\zps)^n$ be code of rank $K$ attaining the join-support bound in Theorem \ref{thm:SB_joinsupp}. Then, for each $r \in \{1, \ldots , K\}$, the $r$-th generalized join-Lee weight is given by
    \begin{align*}
        \LDjoin^r(\code) = p^{s-1}(n-K+r).
    \end{align*}
\end{proposition}
\begin{proof}
    Recall that $\code$ is optimal with respect to the join-support Lee-metric Singleton bound, if its minimum Lee weight code word has the form 
    \begin{align*}
        c_{\min} = 
        \left(
        \begin{array}{ccccccc|ccc}
            0 & \cdots & 0 & p^{s-1} & 0 & \cdots & 0 & \pm p^{s-1} & \cdots & \pm p^{s-1}
        \end{array}    
        \right).
    \end{align*}
    Thus, it holds $\LDjoin^1(\code) = p^{s-1}(n - K + 1)$.

    For arbitrary $r\in \set{1, \ldots , K}$, the $r$-dimensional subcodes of $\code$ attaining the $r$-th generalized join-Lee weight are contained in the socle of the code as well. Hence, they admit a generator matrix $G_r$ which is permutation equivalent to a matrix of the form
    \begin{align*}
        G_r =
        \left(
        \begin{array}{c|c|c}
            \mathbf{0} & p^{s-1} \mathbb{I}_r & \pm p^{s-1} \mathbb{J}
        \end{array} 
        \right),
    \end{align*}
    where $\mathbf{0}$ is the all-zero matrix of size $r\times (K-r)$ and $\mathbb{J}$ is an all-one matrix of size $r\times (n-K)$. Thus, the desired result follows.
\end{proof}

\subsection{Invariance under Isometry in the Lee Metric}\hfill\\
For the generalized Hamming weights of a linear  $k$-dimensional code $\mathcal{C} \subseteq \mathbb{F}_q^n$, we also know that $\HD^r(\code) = \HD^r(\code'),$ for any equivalent code $\code'$ and any $r \in \{1, \ldots, k\}.$ 
Also for the generalized join-Lee weights we have the same behaviour.
\begin{proposition}
    Let $\mathcal{C} \subseteq \left( \mathbb{Z}/ p^s\mathbb{Z}\right)^n$ be a linear code of rank $K$, then 
    $\LDjoin^r(\code)= \LDjoin^r(\code'),$ for all $r \in \{1, \ldots, K\}$ and all $\code'$ which are equivalent to $\code,$ under the Lee-metric isometries.
\end{proposition}
\begin{proof}
    Recall that the Lee-metric isometries only consist of permuting the positions and multiplying any position by $1$ or $-1.$ Thus,   all codewords of $\code'$ can be written as $c' =\sigma(c) \star v$, for some permutation $\sigma$ and $v \in \{1,-1\}^n$, where $\star$ denotes the coordinatewise multiplication and $c \in \code$. Now the claim follows immediately as 
    \begin{align*}
    \LDjoin^r(\code) &= \min\{ |( \max\limits_{c \in \mathcal{C}} \{ \lweight(c_1) \}, \ldots, \max\limits_{c \in \mathcal{C}} \{ \lweight(c_n) \} )| \mid c \in \mathcal{D} \leq \mathcal{C}, \text{rk}(\mathcal{D})=r\} \\
    &= \min\{ |\sigma( \max\limits_{c \in \mathcal{C}} \{ \lweight(c_1) \}, \ldots, \max\limits_{c \in \mathcal{C}} \{ \lweight(c_n\})| \mid c \in \mathcal{D} \leq \mathcal{C}, \text{rk}(\mathcal{D})=r\} \\ & = \LDjoin^r(\code').
    \end{align*}
     
\end{proof}

\section{Column Support for the Lee Metric}\label{sec:colsupport}
We observe that in order to compute the $r$-th generalized Hamming weight of a code $\code$, all we do is considering a generator matrix $G$ and count the number of nonzero columns, i.e., the column weight. In fact, for any $r$-th generalized Hamming weight one can choose $r$ rows of $G$ which attain the minimal column weight.\\

 For a matrix $A \in \mathcal{R}^{K \times n}$ we will denote by $S_r(A) \in \mathcal{R}^{r \times n}$ all the submatrices of $A$ of size $r \times n.$

\begin{definition}
Consider a matrix $A = (a_1^\top \cdots a_n^\top) \in \mathcal{R}^{K \times n}$. We define the \textit{column weight}, $\wt_{\mathrm{C}}(A)$, of $A$ by the number of nonzero columns of $A$, i.e.,
\begin{align*}
    \wtC(A) := \card{\set{ i \in \set{1, \ldots , n} \st a_i \neq 0 \in \mathcal{R}^K}}.
\end{align*}
The \textit{column support}, $\supp_{\mathrm{C}}(A)$, of $A$ is given by
$$\supp_\mathrm{C}(A) := (\max(\supp(a_1)), \ldots, \max(\supp(a_n))).$$ 
\end{definition}
Again we have the nice property that $\card{\supp_\mathrm{C}(A)}= \wtC(A).$
In fact,
$$\wtC(A) = \card{\supp_\mathrm{C}(A) }= \sum_{i=1}^n \max(\supp(a_i)). $$
Thus, we can define the column support, column weight  and the generalized column weights of a code.
\begin{definition}
    Let $\code \subseteq \mathcal{R}^n$ be a linear code of rank $K$. The \textit{column support} of $\code$ is given by the minimal column support of any generator matrix, i.e.,
    \begin{align}\label{def:colsupport_code}
        \supp_\mathrm{C}(\code)=\min_{G: \langle G \rangle = \code} \supp_\mathrm{C}(G).
    \end{align}
    The \textit{column weight} of a code is then given by the size of the column support, i.e.,
    $$\wtC(\code)= \card{\supp_\mathrm{C}(\code)}. $$ 
    Finally, the \textit{$r$-th generalized column weight} of $\code$ is defined as 
    \begin{align}\label{def:gen_colWeight}
        \distC^r(\code)= \min\{\wtC(\subcode) \mid   \subcode \leq \code, \rk(\subcode)=r\}.
    \end{align}
\end{definition}

Note that the definition of the $r$-th generalized column weight of a linear code $\code \subset \mathcal{R}^n$ of rank $K$ is equivalent to 
\begin{align*}
    \distC^r(\code)= \min \{\wtC(S_r(G)) \st \rk(\langle S_r(G) \rangle )=r,  \langle G \rangle = \code\}.
\end{align*}
The difficulty of this new definition lies in the choice of the generator matrix instead of the choice of the subcode. This is the only difference to the usual definition of join support and weight. 


The difficulty of finding the correct generator matrix to read of the minimal column weights, or the subcode with minimal weight is equivalent.\\

Let us show the dependency on the choice of generator matrix in the following example. 
\begin{example}
Let us consider $\code \subseteq \mathbb{F}_2^5$ generated by
\begin{align*}
    G = \begin{pmatrix}
            1 & 0 & 0 & 1 & 1 \\ 
            0 & 1 & 0 & 1 & 1 \\
            0 & 0 & 1& 1 & 1
        \end{pmatrix}.
\end{align*}
If we were to compute the column (Hamming) weights of $S_r(G)$, we would get for $S_1(G)$
$$\wtC\left(\begin{pmatrix} 1 &0 & 0 & 1 & 1 \end{pmatrix}\right)= 3.$$ 
However, this is not the first generalized Hamming weight of the code. 
There exists a generator matrix $G'$,  such that $S_r(G')$  attains the $r$-th generalized Hamming weights as column weights, for each $r \in \{1, \ldots, k\}$:
$$G'= \begin{pmatrix} 
1 & 1 & 0 & 0 & 0 \\
0 &1 & 1 & 0 & 0 \\
0 & 1 & 0 & 1 & 1 
\end{pmatrix}.$$
Now we can read of the $r$-th generalized Hamming weights easily:
\begin{align*}
    \distC^1(\code) & = \wtC\left(\begin{pmatrix} 1 &  1&  0&  0&  0 \end{pmatrix}\right)=2, \\ 
    \distC^2(\code) & = \wtC\left(\begin{pmatrix} 1 &  1&  0&  0&  0 \\ 0 & 1 & 1 & 0 & 0 \end{pmatrix} \right)=3, \\ 
    \distC^3(\code) & = \wtC\left(\begin{pmatrix} 1 &  1&  0&  0&  0 \\ 0 & 1 & 1 & 0 & 0  \\ 0 & 1 & 0 & 1 & 1 \end{pmatrix} \right)=5.
\end{align*}
\end{example}
Thus, the definition  is not independent on the choice of generator matrix.
Let us now adapt the definitions to the Lee weight.

\begin{definition}
Consider a matrix $A = \begin{pmatrix} a_1^\top & \cdots & a_n^\top \end{pmatrix} \in \mathcal{R}^{K \times n}$. Its \textit{column Lee support} is given by the $n$-tuple 
 $$\suppLC(A) = (\max(\suppL(a_1), \ldots , \max(\suppL(a_n)).$$
The \textit{column Lee weight} of $A$ is given by $$\wtLC(A) = \card{\suppLC(A) }= \sum_{i=1}^n \max(\suppL(a_i)). $$
\end{definition}
Note that this definition asks us to choose in each column the entry of maximal Lee weight.

\begin{example}
Let us consider the matrix  
$$G= \begin{pmatrix} 1 & 0 & 3 & 2 \\ 0 & 1 & 2 & 0 \\ 0 & 0 & 3 & 3 \end{pmatrix} \in  \mathbb{Z}/9\mathbb{Z}^{3 \times 4}.$$ 
Then, the column Lee support and the column Lee weight of $G$ are given by
$$\suppLC(G)= (1,1,3,3) \text{ and } \wtLC(G) = 8.$$
\end{example}

We are now able to extend the definitions of column Lee support and column Lee weight to a linear code $\code \subseteq \left(\mathbb{Z}/p^s\mathbb{Z}\right)^n$ of rank $K$.
\begin{definition}\label{def:columnweight}
    Consider a linear code $\code \subseteq \left(\mathbb{Z}/p^s\mathbb{Z}\right)^n$ of rank $K$. We define its \textit{column Lee support} by the minimal column Lee weight of any generator matrix of $\code$, i.e.,
    $$\suppLC(\code)=\min_{G: \langle G \rangle = \code} \suppLC(G).$$
    The \textit{column Lee weight} of $\code$ is then given by the size of its column Lee support, i.e.,
    $$\wtLC(\code)= \card{\suppLC(\code)}.$$
\end{definition}

As in the case for the Hamming metric, also in this case the definition is not independent on the choice of generator matrix. For this, we introduce the following matrix, called \emph{reduced systematic} generator matrix.
\begin{definition}\label{def:reducedSystematicFormG}
    Consider a matrix $G \in \left( \zps \right)^{K \times n}$ as given in \eqref{systematicformG}. We call $G$ to be in \textit{reduced systematic form} if for every entry $a$ of $A_{i, j} \in (\mathbb{Z} / p^{s+1-i} \mathbb{Z})^{k_i \times k_j}$ with $i < j \leq s$ it holds that $\LW(a) \leq p^{j-1}$.
\end{definition}
We will denote a matrix $G$ in reduced systematic form by $G_{\mathsf{rsys}}$. Let us give an example to clarify Definition \ref{def:reducedSystematicFormG}.
\begin{example}
    Consider  $G \in \Zp{27}^{3\times 4}$ 
    \begin{align*}
        G = \begin{pmatrix}
                1 & 14 & 11 & 0 \\ 0 & 9 & 18 & 0 \\ 0 & 0 & 9 & 18
            \end{pmatrix}.
    \end{align*}
    Note that $G$ is in systematic form as defined in \eqref{systematicformG}. By elementary row reduction, i.e., by subtracting suitable multiples of the rows $r_j$ from row $r_i$ with $1 \leq i < j \leq 3$, we obtain a matrix $G_{\mathsf{rsys}}$ in reduced systematic form
    \begin{align*}
        G_{\mathsf{rsys}} = \begin{pmatrix}
                1 & 5 & -7 & 0 \\ 0 & 9 & 9 & -18 \\ 0 & 0 & 9 & 18
            \end{pmatrix}.
    \end{align*}
\end{example}

By a similar argument used to prove Proposition \ref{prop:systematicG} we observe the following.
\begin{proposition}\label{prop:redsystematicG}
    Consider a linear code $\code \subseteq \left(\zps\right)^n$ of rank $K$ and subtype $(k_0, \ldots , k_{s-1})$. The code $\code$ is permutation equivalent to a code having a generator matrix in reduced systematic form.
\end{proposition}

This new systematic form yields a natural upper bound on the column Lee weight of a code $\code$.
For this let us now consider the support subtype outside an information set of size $K$ of the code. Since we can always find a permutation equivalent code, which has an information set in the first $K$ positions, we can assume that we only consider the  last $n-K$ columns of a generator matrix in reduced systematic form.  In order not to confuse it with the support subtype $(n_0, \ldots, n_s)$ of the entire generator matrix, we will denote it by $(\mu_0, \ldots , \mu_s)$. 
\begin{proposition}\label{prop:bound_weightC}
       Let $\code \subseteq (\zps)^n$ be a linear code of rank $K$ and subtype $(k_0, \ldots, k_{s-1})$ and let $(\mu_0, \ldots, \mu_{s-1})$ be the support subtypes in the last $n-K$ positions. 
       Then the column Lee weight of $\code$ is upper bounded by
       \begin{align*}
           \wtLC (\code) \leq \sum_{i = 0}^{s-1} p^ik_i + \sum_{i = 0}^{s}\mu_iM_i.
       \end{align*}
\end{proposition}
\begin{proof}
    By Definition \ref{def:columnweight} we have 
    \begin{align*}
        \wtLC (\code) = \card{ \min_{G: \langle G \rangle = \code} \suppLC(G) }.
    \end{align*}
    Furthermore, by Proposition \ref{prop:redsystematicG}, $\code$ admits a generator matrix $\Grsys \in (\zps)^{K \times n}$ in reduced systematic form. Hence, the column Lee weight of $\Grsys$ is a natural upper bound to the column Lee weight of the code, i.e.,
    \begin{align*}
        \wtLC (\code) \leq \wtLC (\Grsys).
    \end{align*}
    Thanks to the form of $\Grsys$, we observe that the maximum Lee weight in the first $K$ columns is given by the entry $(\Grsys)_{i,i}$ for $i \in \{1, \ldots, K\}$. For the last $n-K$ columns we have to assume the maximal Lee weight. The support subtype $(\mu_0. \ldots, \mu_{s})$ in these positions immediately tells us, how many columns are contained in which  ideal. Hence, for each column lying in $\langle p^i  \rangle$ (where $i$ is maximal for this column) the maximal Lee weight is $M_i$. This yields the desired result.
\end{proof}

Let us now introduce the $r$-th generalized column Lee weights of a code $\code$. 
\begin{definition}
    Given a linear code $\code \subseteq (\zps)^n$ of rank $K$ and subtype $(k_0, \ldots , k_{s-1})$. The \textit{$r$-th generalized column Lee weight} of $\code$ is defined as
    \begin{align}\label{def:genLee_colWeight}
        \distLC^r(\code)= \min\{\wtLC(\subcode) \st \subcode \leq \code, \rk(\subcode)=r\}.
    \end{align}
\end{definition}

Similarly to Definition \eqref{def:gen_colWeight}, the $r$-th generalized column Lee weight is equivalent to
$$\distLC^r(\code)= \min \{\wtLC(S_r(G)) \mid  \rk(\langle S_r(G) \rangle )=r, \langle G \rangle = \code\}.$$

As in the Hamming-metric case, the difficulty lies now in finding a generator matrix attaining the $r$-th generalized column Lee weights. 
To visualize this, let us return to our previous example for the Lee-metric support.
\begin{example}\label{ex:favExample}
    Let us consider the code $\code \subseteq \mathbb{Z}/9\mathbb{Z}^4$ generated by 
    $$G= \begin{pmatrix} 1 & 0 & 3 & 2 \\ 0 & 1 & 2 & 0 \\ 0 & 0 & 3 & 3 \end{pmatrix},$$ which has support subtype $(4,0,0)$ and minimum Lee distance $2$.
    
    If we compute the minimal column weights of submatrices of  $G$ we get 
    \begin{align*} 
        \wtLC(\begin{pmatrix} 0 & 1 & 2 & 0 \end{pmatrix}) =3, \\ 
        \wtLC\left(\begin{pmatrix} 0 & 1 & 2 & 0 \\ 0 & 0 & 3 & 3  \end{pmatrix}\right) =7, \\ 
        \wtLC(G)=8.
     \end{align*}
     
     However, there is a generator matrix of the code which is not in systematic form and which attains smaller column Lee weights: 
     $$G'=\begin{pmatrix}
     8 & 0 & 0 & 1 \\
     0 & 1 & 2 & 0 \\
     0  & 8 & 1 & 3 \end{pmatrix}.$$
     The $r$-th generalized Lee weights are then
     \begin{align*} 
        \distLC^1(\code) &= \wtLC(\begin{pmatrix} 8 & 0 &0 & 1 \end{pmatrix})=2 = \LD(\code), \\ 
         \distLC^2(\code) &= \wtLC\left(\begin{pmatrix} 8 & 0 & 0 & 1  \\ 0  & 1 & 8 & 0  \end{pmatrix}\right)=4, \\ 
         \distLC^3(\code) &= \wtLC \left(\begin{pmatrix} 8 & 0 & 0 & 1  \\ 0  & 1 & 8 & 0\\ 0  & 0 & 3 & 0  \end{pmatrix}\right) = 6 = \wtLC(\code).
     \end{align*}
     Note that both matrices within this example are of reduced systematic form.
\end{example}

\begin{lemma}\label{lem:it}
    Let $\code \in \zpsn$ be code of rank $K$. 
    Let $G^{(i)} \in \left(\mathbb{Z}/p^s\mathbb{Z}\right)^{i \times n}$ of a rank $i \in \{1, \ldots, K-1\}$ be a generator matrix of a subcode of $\code$ attaining $\distLC^i(\code)$. Let $c \in \code$ such that $\begin{pmatrix} G^{(i)} \\ c \end{pmatrix}$ is of a generator matrix of a subcode of rank $i+1$. Then it holds, 
    $$\wtLC\left(  \begin{pmatrix} G^{(i)} \\ c \end{pmatrix} \right) > \wtLC (G^{(i)}). $$
\end{lemma}
\begin{proof}
    Let us define for all columns $j \in \{1, \ldots, n\}$ the maximal Lee weight of the $j$-th column in $G^{(i)}$ as $A_j^{(i)}$. 
    We clearly have $\wtLC(G^{(i)}) \leq \wtLC\left( \begin{pmatrix} G^{(i)} \\ c \end{pmatrix}\right).$ 
    Thus, let us assume that $\wtLC(G^{(i)}) = \wtLC\left( \begin{pmatrix} G^{(i)} \\ c \end{pmatrix}\right)$.
    Then 
    $$\sum_{j=1}^n A_j^{(i)} = \sum_{j=1}^n \max\{ A_j^{(i)}, \lweight(c_j) \}$$
    and so for all $j \in \{1, \ldots, n\}$ we have $\lweight(c_j) \leq A_j^{(i)}$.
    However, as $G^{(i)}$ attains $\distLC^i(\code)$, the sum $\sum_{j = 1}^n A_j^{(i)}$ is minimal among all rank $i$ subcodes of $\code$. Hence, there is no index $j\in\set{1, \ldots, n}$ for which $\lweight(c_j) < A_j^{(i)}$ and thus for all $j$ we have $\lweight(c_j)=A_j^{(i)}$. This implies $c_j = \pm A_j^{(i)}.$
    
    This means that $c$ has in every position the maximal Lee weight over all rows of $G^{(i)}$.
    Thus, for every row $g_\ell$ of $G^{(i)}$ with $\ell \in \set{1, \ldots, i}$  for which $\LW(g_{\ell}) > \LW(c)$, we can add and/or subtract $c$ to decrease its weight.
    For each row $\ell \in \set{1, \ldots , i}$ let us therefore define the sets 
    \begin{align*} 
        I_\ell^- &= \{ j \in \{1, \ldots, n \} \mid  \lweight(c_j - g_{\ell j}) < \lweight(c_j)\},  \\
        I_\ell^+ &= \{j \in \{1, \ldots, n\} \mid \lweight(c_j + g_{\ell j}) \leq \lweight(c_j)\}.
    \end{align*}

    For a fixed row $\ell \in \{1, \ldots, i\}$, if 
    $$\sum_{j \in I_\ell^-} \lweight(c_j)< \sum_{j \in I_\ell^+} \lweight(c_j), $$ 
    then we add $c$ to the row $g_\ell$. If however, 
    $$\sum_{j \in I_\ell^+} \lweight(c_j) \leq \sum_{j \in I_\ell^-} \lweight(c_j),$$ 
    we subtract $c$ from that row $g_\ell.$
    
    We consider now the new row $g'_{\ell} := c \pm g_\ell $ which has a strictly smaller Lee weight than $c$. 
    Since the cases are similar, assume that for $g_{\ell}$ the first case is true, i.e., $\sum_{j \in I_{\ell}^-} \lweight(c_j) < \sum_{j \in I_{\ell}^+} \lweight(c_j)$ and thus we add the row $c$, getting $g'_{\ell}:= c+g_{\ell}$. 
    Clearly for each position $j$ in $I_{\ell}^-$ we at most added a Lee weight of $A_j^{(i)}$, while in each position $j$ in $I_{\ell}^+$ we subtracted a Lee weight of at most $A_j^{(i)}$,  thus 
    \begin{align*} 
        \lweight(g'_{\ell}) 
        &=\sum_{j \in I_{\ell}^+} \lweight(g_{\ell j} + c_j) + \sum_{j \in I_{\ell}^-} \lweight(g_{\ell j}+c_j) \\   
        &< \sum_{j \in I_{\ell}^+} \lweight(g_{\ell j} + c_j) + \sum_{j \in I_{\ell}^-} A_j^{(i)} +  \sum_{j \in I_{\ell}^-} \lweight((c_j) \\ 
        & <\sum_{j \in I_{\ell}^+} \lweight(c_j) - \sum_{j \in I_{\ell}^+} A_j^{(i)} 
        + \sum_{j \in I_{\ell}^+} A_j^{(i)} +  \sum_{j \in I_{\ell}^-} \lweight((c_j)\\
        &= \lweight(c).
    \end{align*}
    Doing this procedure for every row of the matrix $G^{(i)}$,  obtaining the new matrix $G'^{(i)}$ of rank $i$, we have $$\wtLC(G'^{(i)}) < \wtLC(G^{(i)}), $$ since in every row we now reduced the Lee weight, but this is a contradiction to $G$ attaining $\distLC^i(\code)$.
\end{proof}

Finally, we are able to prove the desired properties for the generalized column Lee weights.

\begin{proposition}\label{prop:colWeight_properties}
    Let $\code \subseteq (\zps)^n$ be a linear code of rank $K$. Then
    \begin{enumerate}
        \item $\distLC^1(\code)=\LD(\code).$
        \item $\distLC^r(\code)< \distLC^{r+1}(\code)$ for all $r<K.$
        \item $\distLC^K(\code)=\wtLC(\code).$
    \end{enumerate}
\end{proposition}
\begin{proof} 
For the first property, note that the column Lee weight of a $1 \times n$ matrix is equal to the Lee weight of that $n$-tuple. Since a minimal Lee-weight codeword $c$ is   the rank 1 subcode of $\code$ with the smallest column Lee weight, it  attains  $\wt_{L,C}(c)=\distLC^1(\code)$.

The second property follows from Lemma \ref{lem:it}.
Since then for any matrix $G^{(i+1)} \in \left( \mathbb{Z}/p^s\mathbb{Z}\right)^{(i+1) \times n}$, which attains $\distLC^{i+1}(\code)$, which we can write as $G^{(i+1)}= \begin{pmatrix} G^{(i)} \\ g' \end{pmatrix}$, we either have that $G'_{i}$ already attained $\distLC^i(\code)$ and hence
$$\distLC^i(\code) = \wtLC(G^{(i)}) < \wtLC(G^{(i+1)})=\distLC^{i+1}(\code),$$ or if $G^{(i)}$ did not attain $\distLC^i(\code)$, then 
$$\distLC^i(\code) < \wtLC(G^{(i)}) \leq \wtLC(G^{(i+1)}).$$
In either case, we get that 
$\distLC^i(\code) < \distLC^{i+1}(\code).$



Lastly, the third property follows immediately from the definition of the column Lee weight of a code $\code$.
\end{proof}

The properties in Proposition \ref{prop:colWeight_properties} allow us to deduce a  natural Singleton-like bound for the Lee metric.
\begin{theorem}\label{thm:SB_wt-k+1}
    Given a linear code $\code \in (\zps)^n$ of rank $K$. The minimum distance of $\code$ is upper bounded by
    \begin{align*}
        \LD(\code) \leq \wtLC(\code) - K + 1.
    \end{align*}
\end{theorem}
\begin{proof}
    Using the properties given in Proposition \ref{prop:colWeight_properties} we note that 
    \begin{align}\label{equ:upperbound_d1}
        \LD(\code) = \distLC^1(\code) \leq \distLC^K(\code) - \sum_{i = 2}^{K} \left( \distLC^i(\code) - \distLC^{i-1}(\code) \right).
    \end{align}
    By the strict inequality between the generalized column Lee weights, we have a difference of at least one, i.e., $$\distLC^i(\code) - \distLC^{i-1}(\code) \geq 1.$$
    Since $\distLC^K(\code) = \wtLC(\code)$, the desired bound follows.
\end{proof}
As for increasing parameters of a linear code $\code \subseteq (\zps)^n$ of rank $K$ and subtype $(k_0, \ldots , k_{s-1})$ it becomes harder to compute $\wtLC(\code)$, applying Proposition \ref{prop:bound_weightC} we obtain a direct consequence to Theorem \ref{thm:SB_wt-k+1}, which requires no computational effort.
\begin{corollary}\label{cor:SB_wt-k+1}
    Given a linear code $\code \in (\zps)^n$ of rank $K$. The minimum distance of $\code$ is upper bounded by
    \begin{align*}
        \LD(\code) \leq\sum_{i = 0}^{s-1} p^ik_i + \sum_{i = 0}^{s}\mu_iM_i - K + 1.
    \end{align*}
\end{corollary}
The bounds given in Theorem \ref{thm:SB_wt-k+1} and Corollary \ref{cor:SB_wt-k+1} improve the Singleton bound by Shiromoto \cite{shiromoto} and the one by Alderson-Huntemann \cite{alderson}. In the proof of Theorem \ref{thm:SB_wt-k+1} we bounded the differences $\distLC^i(\code) - \distLC^{i-1}(\code)$ by one for every $i = 2, \ldots , K$. However, for a relatively small rank $K$ this bound is not very tight. The sum in Equation \eqref{equ:upperbound_d1} is a telescoping sum, meaning that 
\begin{align*}
    \sum_{i = 2}^{K} \left( \distLC^i(\code) - \distLC^{i-1}(\code) \right) = \distLC^K(\code) - \distLC^1(\code) = \wtLC(\code) - \LD(\code).
\end{align*}
Hence, the goal is now to derive a lower bound on the difference  $\wtLC(\code) - \LD(\code)$ allowing us to further tighten the Singleton-like bound.\\ 

In the following let $\code \subseteq (\zps)^n$  be a linear code of rank $K$ and subtype $(k_0, \ldots , k_{s-1})$. Let us introduce the maximal subtype $i \in \set{0, \ldots, s-1}$ for which $k_i$ is nonzero, that is
\begin{align*}
    \sigma := \max\set{ i \in \set{0, \ldots, s-1} \st k_i \neq 0}.
\end{align*}

\begin{proposition}\label{prop:lowerbound_difference}
    Let $p$ be an odd prime. For a linear code $\code \subseteq (\zps)^n$ of rank $K$ and subtype $(k_0, \ldots , k_{s-1})$ and maximal subtype $k_\sigma$, we get the following lower bound
    \begin{align*}
        \wtLC(\code) - \LD(\code) \geq \sum_{i = 0}^{\sigma - 1} \left( \sum_{j = 0}^{i} k_j \right)\floor{p/2}p^{i} + (k_{\sigma}-1)p^{\sigma}.
    \end{align*}
\end{proposition}
\begin{proof}
    Let us start by focusing on the generalized column Lee weights. Assume that $c_1 \in \code$ is such that $\distLC^1(\code) = \wtLC(\langle c_1 \rangle)$. By Lemma \ref{lem:it}, we know that the generalized column Lee weights can be obtained in an iterative fashion. Hence, to find a subcode $\subcode_2$ of rank $2$ we are looking for a codeword $c_2 \in \code$ such that $\begin{pmatrix}c_1 \\c_2\end{pmatrix}$ is of rank $2$ and such that it minimizes $\wtLC\left(\begin{pmatrix} c_1 \\c_2 \end{pmatrix}\right)$. We continue with this process until we obtain a matrix
    \begin{align*}
        G_K := \begin{pmatrix} c_1 \\ \vdots \\ c_K \end{pmatrix}
    \end{align*}
    of rank $K$ such that $\wtLC(G_K) = \wtLC(\code) = \distLC^{K}(\code)$.

    Since the code $\code$ is of subtype $(k_0, \ldots , k_{\sigma}, 0,\ldots, 0)$, the rows of the matrix $G_K$ each correspond to one of the $\sigma$ blocks formed by the systematic form $\Gsys$ of $G_K$. To understand the difference of $\wtLC(\code)$ and the first generalized column Lee weight $\distLC^1(\code)$ we can think of successively removing rows from $G_K$ until we are only left with the minimum weight codeword $c_1$. Thinking in the block-wise structure of a generator matrix in systematic form, at some point we will have cancelled $k_i$ rows corresponding to the $i$-th block of $\Gsys$. Hence, the minimal difference subtracted is $$M_{i-1} - M_{i} = \floor{p/2}p^{i-1}.$$
    Doing this successively for every $k_i$, with $i \in \{ 0, \ldots, \sigma \}$, gives
    $$\sum_{i = 0}^{\sigma - 1}\left(\sum_{j = 0}^i k_j\right) \floor{p/2} p^{i}.$$
    At this point we are left only with a block corresponding to the rows belonging to the maximal subtype $k_{\sigma}$. The minimal difference between the rows of the same block is given by $p^\sigma$. Hence, cancelling $(k_\sigma - 1)$ rows yields to a difference of $p^{\sigma}(k_{\sigma} -1)$ and the desired result follows.
\end{proof}
A natural consequence (combining Propositions \ref{prop:bound_weightC} and \ref{prop:lowerbound_difference}) is the next bound on the minimum Lee distance $\LD(\code)$ of a code $\code$ of given rank and subtype.
\begin{corollary}\label{cor:SB_column_cancelling}
    Consider a linear code $\code \subseteq (\zps)^n$, where $p$ is an odd prime. Let $\code$ be of rank $K$ and subtype $(k_0, \ldots , k_{s-1})$ with maximal subtype $k_\sigma$ and having support subtype $(\mu_0, \ldots, \mu_{s-1})$ in the last $n-K$ positions. Then the following upper bound on the minimum Lee distance of $\code$ holds
    \begin{align*}
        \LD(\code) \leq \sum_{i = 0}^{s-1} p^ik_i + \sum_{i = 0}^{s}\mu_iM_i - \left[  \sum_{i = 0}^{\sigma - 1} \left( \sum_{j = 0}^{i} k_j \right)\floor{p/2}p^{i} + (k_{\sigma}-1)p^{\sigma} \right].
    \end{align*}
\end{corollary}
Let us give an example over $\Zp{9}$.

\begin{example}\label{ex:favExample_Z9}
    Consider again the code $\code$ generated by 
    $$G= \begin{pmatrix} 1 & 0 & 3 & 2 \\ 0 & 1 & 2 & 0 \\ 0 & 0 & 3 & 3 \end{pmatrix},$$ over $\Zp{9}$. In the last column, the code $\code$ has support subtype $(1, 0, 0)$ and minimum Lee distance $2$. Furthermore, we observe that $\sigma = 1$ and support subtype $(1,0)$. Hence, by Corollary \ref{cor:SB_column_cancelling}
    \begin{align*}
        \LD(\code) \leq 2 + 3 + 1\cdot4 - \left[ 2\cdot 1 + (1 - 1)3 \right] = 7.
    \end{align*}
\end{example}

Similarly to the join-support, examples of codes attaining this bound are codes generated by matrices $G = \begin{pmatrix} p^{s-1}\mathbb{I}_K & p^{s-1}A \end{pmatrix}$ for $A \in (\zps)^{K \times (n-K)}$, where $p = 3$. In fact, for any odd $p$ these codes have a minimum Lee distance $d = p^{s-1}(n-K+1)$. Furthermore we note that in the last $n-K$ positions we have support subtype  $(0, \ldots, 0, n-K)$ and $M_{s-1} = \floor{p/2} p^{s-1}$. Hence, inserting these values in the bound given in Corollary \ref{cor:SB_column_cancelling} gives
$$\LD(\code) \leq p^{s-1}(1 + (n-K) \floor{p/2}).$$
This is equal to $\LD(\code)$ exactly if $p = 3$.\\
For instance, consider again the code $\code \subseteq (\Zp{9}^4)$ of rank $K = 3$ with generator matrix
\begin{align*}
    G = 
    \begin{pmatrix}
        3 & 0 & 0 & 3 \\
        0 & 3 & 0 & 6 \\
        0 & 0 & 3 & 6
    \end{pmatrix}.
\end{align*}
This code has minimum Lee distance $\LD(\code) = 6$ and subtype $(k_0, k_1) = (0, 3)$. Hence, we also have $\sigma = 1$. The support subtype in the last $n-K = 1$ positions is $ (0, 1, 0)$ and $M_1 = 3$. Computing the bound in Corollary \ref{cor:SB_column_cancelling} gives then
\begin{align*}
    \sum_{i = 0}^{1} p^ik_i + \sum_{i = 0}^{2}\mu_iM_i -  (k_{1}-1)p
        =
    3 \cdot 3 + 1\cdot 3 - (3-1)3 = 2\cdot 3 = 6
\end{align*}
and we conclude that this code is optimal with respect to Lee-metric Singleton bound \ref{cor:SB_column_cancelling}.

\subsection{Density of Optimal Codes with respect to the Column-Lee Support}  
Let us discuss the density of the codes attaining the bound in Corollary \ref{cor:SB_column_cancelling}. Recall that the bound is derived by $\LD(\code) \leq \wtLC(\code) - (\wtLC (\code) - \LD(\code))$ where we upper bounded the column weight of the code by $\wtLC(\code) \leq \sum_{i = 0}^{s-1} p^ik_i + \sum_{i = 0}^{s}\mu_iM_i$. Hence, in order to have codes attaining the bound on the minimum Lee distance, they must attain the bound on the column Lee weight too. That is, their generator matrix $G$ must be in reduced systematic form. Furthermore, the support subtype of the last $n-K$ positions is $(\mu_0, \ldots, \mu_s)$ where in each of the $\mu_i$ positions the maximum Lee weight $M_i$ is attained. For instance, a generator matrix may look as follows:
\begin{align*}
    \Grsys = 
    \begin{bNiceArray}{cw{c}{1cm}c|cccc}[margin, last-row = 4]
        \Block{3-3}{U} & & & & & & \\
            & & & \hspace*{1.2cm} & \hspace*{0.4cm} & \ldots & \hspace*{0.6cm} \\
            & & & \hspace*{1.2cm} & & & \\
            & & & {\color{RoyalBlue}\mu_0} & {\color{YellowOrange}\mu_1} & & {\color{myred}\mu_s}
        \CodeAfter
            \color{RoyalBlue}\SubMatrix[{1-4}{3-4}]
            \color{YellowOrange}\SubMatrix[{1-5}{3-5}]
            \color{myred}\SubMatrix[{1-7}{3-7}]
    \end{bNiceArray}
\end{align*}
There are two options to attain a Lee weight $M_i$. Hence, the probability that a generator matrix is of this form is given by the number of such matrices divided by the number of all matrices, i.e.,
\begin{align*}
    \prod_{i = 0}^{s-1} \left(\frac{2 (p^{s-i})^{(k-1)} }{(p^{s-i})^{(k-1)}(p^{s-i}-p^{s-i-1})} \right)^{\mu_i} 
    &= \prod_{i = 0}^{s-1} \left(\frac{2 }{p^{s-i}-p^{s-i-1}} \right)^{\mu_i}.\\
    &= 2^{n-K}\prod_{i = 0}^{s-1} \left( \frac{1}{p^{s-i}(1- 1/p)} \right)^{\mu_i} \\
    &= 2^{n-K}\prod_{i = 0}^{s-1} \left( \frac{p^{i+1}}{p^s(p-1)} \right)^{\mu_i}.
\end{align*}
Note that $p^s(p-1) > p^{i+1}$ for every $i \in \{0, \ldots, s-1\}$. Hence, the fraction in the product is smaller than $1$. Therefore, for $p \tendsto \infty$ the product tends to $0$. The same argument holds if we let $s$ tend to infinity. Similarly, as $\mu_i$ depends on $n$, we note  that $\frac{2}{p^{s-i}-p^{s-i-1}} < 1$. This implies that if $n \longrightarrow\infty$ the product tends to zero as well. Thus, codes attaining the bound in Corollary \ref{cor:SB_column_cancelling} are sparse with respect to $p$, $s$ and $n$.

Given an optimal code with respect to the Lee-metric Singleton bound \ref{cor:SB_column_cancelling}, one could also ask if the $r$-th generalized column Lee weights are then also fixed. Since the main problem of the column Lee weight of a code is the computational difficulty, we leave this as an open question.

\subsection{Invariance under Isometry in the Lee Metric}\hfill\\
Finally, we ask if the $r$-th generalized column Lee weights are fixed under isometries. 
\begin{proposition}
    Let $\mathcal{C} \subseteq \left( \mathbb{Z}/p^s\mathbb{Z}\right)^n$ be a linear code of rank $K$, then any equivalent code $\code'$, under the linear Lee-metric isometries is such that
     $$ \distLC^r(\code)= \distLC^r(\code'),$$ for every $r \in \{1, \ldots, K\}$.
\end{proposition}

    \begin{proof}
   Recall that any generator matrix $G'^{(i)}$ of a subcode of rank $i$  of a equivalent $\code'$ can be written as $G'^{(i)} =G^{(i)}P \text{diag}(v)  $, for some permutation matrix $P$, $v \in \{1,-1\}^n$ and some generator matrix $G^{(i)}$ of a subcode of rank $i$ of $\code$. Both, $G'^{(i)}$ and $G^{(i)}$ have the same column weight. Now the claim follows immediately as 
    \begin{align*}
    \distLC^r(\code) &= \min\{ \wtLC(G^{(r)}) \mid \langle G^{(r)} \rangle \leq \mathcal{C}, \text{rk}(\langle G^{(r)} \rangle)=r\} \\
    &= \min\{ \wtLC(G^{(r)}P\text{diag}(v)) \mid \langle G^{(r)}  \rangle \leq \mathcal{C}, \text{rk}(\langle G^{(r)}  \rangle)=r\} = \distLC^r(\code').
    \end{align*}
\end{proof}

\section{Generalized Lee Weights from Filtration}\label{sec:genWeight_filtration}

The resulting Lee-metric Singleton bounds in Theorem \ref{thm:SB_joinsupp} and Corollary \ref{cor:SB_column_cancelling} are improving the previously known bounds, however their optimal codes are sparse and the column Lee weight of a code is computationally difficult to compute. 
We thus ask if fixing the rank of the subcode is the correct direction. In fact, a ring-linear code $\mathcal{C} \subseteq \left(\mathbb{Z}/p^s\mathbb{Z}\right)^n$ of rank $K$ has very natural subcodes to consider, which are all of rank $K.$

\begin{definition}
    For each $i \in \set{0, \ldots, s-1}$ we define the \textit{$i$-th filtration} subcode $\code_i$ of $\code$ as the intersection of $\code$ with the ideal $\langle p^i \rangle$, i.e., 
    $$ \code_i := \code \cap \langle p^i \rangle. $$
    The $(s-1)$-st filtration $\code_{s-1}$ is commonly known as the \textit{socle} of the code $\code$.
\end{definition}
Note that the filtration subcodes naturally form a chain, namely
\begin{align}\label{eq:inclusion_filtrations}
    \code_{s-1} \subseteq \code_{s-2} \subseteq \ldots \subseteq \code_1 \subseteq \code_0 = \code.
\end{align}
We then define a new class of generalized Lee weights, or more concretely generalized Lee distances, coming from filtration subcodes.
\begin{definition}\label{def:genLD}
    Let $\code \subseteq \left( \zps \right)^n$ be a linear code. 
    For each $r \in \set{1, \ldots, s}$ we define the \textit{$r$-th generalized minimum Lee distance} of the code $\code$ to be the minimum distance of the filtration  subcode $\code_{r-1}$, that is
    \begin{align*}
        \LD^r(\code) = \LD(\code_{r-1}).
    \end{align*}
\end{definition}
The generalized minimum Lee distances have some natural properties that are summarized in the following.
\begin{proposition}\label{prop:genLD_properties}
    Given a linear code $\code \subseteq \left( \zps \right)^n$ of rank $K$ and subtype $(k_0, \ldots, k_{s-1})$, let $\sigma := \max\set{ i \in \set{0, \ldots, s-1} \st k_i \neq 0}$. Then the generalized minimum Lee distances satisfy
    \begin{enumerate}
        \item $\LD^1(\code) = \LD(\code)$, \vspace{1mm}
        \item $\LD^r(\code) \leq \LD^{r+1}(\code)$ for every $r \in \set{1, \ldots s-1}$,\vspace{1mm}
        \item $\LD^r(\code) \leq p^{r-1} + (n-k)M_{r-1} $ for every $r = \sigma + 1, \ldots , s$.\label{item:bound_minLeedist}
    \end{enumerate}
\end{proposition}
\begin{proof}
    The first and second property immediately follow from \eqref{eq:inclusion_filtrations}.\\
    For the third property we observe that for every $r \in \set{\sigma + 1, \ldots , s}$,   by applying elementary row operations, we can bring a generator matrix $G_{r-1}$ of $\code_{r-1}$ in the form
    \begin{align}\label{eq:G_filtration}
        G_{r-1} =
        \begin{pmatrix}
            p^{r-1}\mathbb{I}_{K} & A
        \end{pmatrix},
    \end{align}
    where $A \in \left(p^{r-1}\zps\right)^{K \times (n-K)}$. The $r$-th minimum Lee distance is upper bounded by the Lee weight of any row of $G_{r-1}$. For each row, the first $K$ positions have a Lee weight of exactly $p^{r-1}$. In the last $n-K$ positions of each row we assume the maximal Lee weight given by $M_{r-1} := \floor{p^{s-(r-1)}/2}p^{r-1}$ and hence the inequality follows.
\end{proof}

Due to Property 2. in Proposition \ref{prop:genLD_properties}, we cannot use the usual Singleton-like argument and decrease the weight of the whole code. Instead, we note that any $\LD^r(\code)$ is a direct upper bound on the minimum Lee distance. The only question that remains, is how far we have to go down in the filtration to expect the lowest minim Lee distance, $\LD^r(\code)$. In the following we identify several parameters of the code, that are easy to read off from any generator matrix of the code, that indicates which filtration subcode gives an appropriately low upper bound on $\LD(\code).$

The upper bound on the generalized minimum Lee distances in Property \ref{item:bound_minLeedist}. of Proposition \ref{prop:genLD_properties} is relatively loose. This is due to the fact, that we have assumed no knowledge about the matrix $A$ given in \eqref{eq:G_filtration}.\\

As computing the minimum Lee distance of every subcode $\code_i$ is an exhausting task, especially if there is no knowledge about the structure of $A$, we would like to introduce some more parameters regarding $A$ for the first filtration of $\code$ admitting a generator matrix of the form \eqref{eq:G_filtration}.
That is the filtration $\code_{\sigma}$ with a generator matrix of the form $G_{\sigma} = \begin{pmatrix} p^{\sigma}\mathbb{I}_{K} & A \end{pmatrix}$, for some matrix $A \in \left(p^{\sigma} \zps \right)^{K \times (n-k)}$.
Let $a_{ij}$ denote the entry of $A$ lying in row $i$ and column $j$.
For each row of $A$, we determine the maximal power of $p$ appearing and we denote it by
\begin{align*}
    \ell_i := \max \set{ k \in \set{\sigma, \ldots , s-1} \st \exists a_{ij} : \langle a_{ij} \rangle = \langle p^k \rangle,\, K+1 \leq j \leq n }.
\end{align*}
Clearly, $\ell_i \geq \sigma.$
Let $n'_i$ denote the number of entires of the $i$-th row  of $A$ that live in the ideal $\langle p^{\ell_i} \rangle$, i.e.,
\begin{align*}
    n'_{\ell_i} := \card{\set{ j \in \set{K+1, \ldots , n} \st a_{ij} \in \langle p^{\ell_i} \rangle }}.
\end{align*}

For a given linear code $\code \subseteq \zpsn$, these parameters help us to understand the evolution of the matrix $A$ in the generator matrices of the filtration subcodes $\code_{r-1}$, for $r \in \{\sigma + 1, \ldots , s\}$. In fact, given a generator matrix $G_{\sigma}$ of the filtration $\code_{\sigma}$ in the form \eqref{eq:G_filtration}, the parameters $\ell_i$ and $n'_{\ell_i}$ for a row $i \in \set{1, \ldots, K}$ allow to understand at which point in the filtration these positions become zero.  More precisely, knowing $\ell_i$ and $n'_{\ell_i}$ implies that in  $\code_{s-\ell_i + \sigma}$ there are $n'_{\ell_i}$ many zero entries in $i$-th row of $A$.

Knowing the number of entries turning into zero in a certain filtration is a huge advantage in bounding the minimum distance of a code. Therefore, we define by $n'^{(r-1)}$ the maximal number of zeros we can get in the last $n-K$ positions of a row of a generator matrix of the filtration $\code_{r-1}$. That is, for every $r \in \set{\sigma + 1, \ldots, s}$,
\begin{align*}
    n'^{(r-1)} := \max\set{ n'_{\ell_i} \st \ell_i > s - r + \sigma,\; i \in \{1, \ldots, K\} }.
\end{align*}
If there is no $\ell_i$ with $\ell_i > s - r + \sigma$, we will set $n'^{(r-1)} = 0$.
Furthermore, let $\ell^{(r-1)}$ be the corresponding value $\ell_i$ to $n'^{(r-1)}$, i.e.,
\begin{align*}
    \ell^{(r-1)} := \max\set{\ell_i \st n'_{\ell_i} = n'^{(r-1)}, \; i \in \{ 1, \ldots, K\}}.
\end{align*}

We can hence refine the third property in Proposition \ref{prop:genLD_properties} as follows.
\begin{lemma}\label{lemma:bound_rth}
    Given a linear code $\code \subseteq \zpsn$ of subtype $(k_0, \ldots, k_{s-1})$ with maximal subtype $k_\sigma$. Then, for every $r \in \set{\sigma +1 , \ldots , s}$, the $r$-th generalized Lee distance can be upper bounded by
    \begin{align*}
        \LD^r(\code) = \LD(\code_{r-1}) \leq p^{r-1} + (n - K - n'^{(r-1)})M_{r-1}.
    \end{align*}
\end{lemma}
\begin{proof}
    The proof follows in a similar fashion as the proof of Proposition \ref{prop:genLD_properties} by focusing on the row with the maximal number of zeros in the last $n-K$ columns of $\code_{r-1}$ which is captured in $n'^{(r-1)}$. Hence, the remaining $(n-K-n'^{(r-1)})$ positions are bounded by the maximal Lee weight in the ideal considering, which is given by $M_{r-1}$.
\end{proof}

\begin{example}\label{ex:ell_n'}
    Let us consider a free code $\code \in (\Zp{27})^5$ spanned by the rows of the matrix
    \begin{align*}
        \begin{pmatrix}
            1 & 0 & 0 & 21 & 6 \\
            0 & 1 & 0 & 10 & 7 \\
            0 & 0 & 1 & 18 & 8
        \end{pmatrix}
        =: 
        \begin{pmatrix}
            \mathbb{I}_3 & A 
        \end{pmatrix}.
    \end{align*}
   We easily check that
    \begin{align*}
        \ell_1 = 1 \quad\text{and}\quad n'_{\ell_1} = 2, \\
        \ell_2 = 0 \quad\text{and}\quad n'_{\ell_2} = 2, \\
        \ell_3 = 2 \quad\text{and}\quad n'_{\ell_3} = 1.
    \end{align*}

    Let us now consider the filtration subcodes $\code_1$ and $\code_2$ in order to compute the bound given in Proposition \ref{prop:genLD_properties}.
    Note that in this case $\sigma = 0$ as the code is free. For $\code_{\sigma} = \code_0 = \code$, the values $\ell_i$ and $n'_i$ are given above. As $\ell_3 = 2$ and $n'_3 = 1$, at filtration $\code_{3-2+0} = \code_1$ there is one entry equal to zero. Indeed, $\code_1 = \code \cap \langle 3 \rangle$ has a generator matrix of the form
    \begin{align*}
        \begin{pmatrix}
            3 & 0 & 0 & 9 & 18 \\
            0 & 3 & 0 & 3 & 21 \\
            0 & 0 & 3 & 0 & 24
        \end{pmatrix},
    \end{align*}
    where the last row contains one zero element in the last $2$ columns.
    Note that $n'^{(1)} = n'_{\ell_3} = 1$ and hence, $\LD^2(\code) \leq 3 + (5-3-1)12 = 15$.
    
    Similarly, at filtration $\code_2 = \code \cap \langle 9 \rangle$ we  observe two zero entries in the first row, as
    \begin{align*}
        \begin{pmatrix}
            9 & 0 & 0 & 0 & 0 \\
            0 & 9 & 0 & 9 & 9 \\
            0 & 0 & 9 & 0 & 18
        \end{pmatrix}.
    \end{align*}
    Here we notice that $n'^{(2)} = n'_{\ell_1} = 2$ and thus $\LD^3(\code) \leq 9 + (5-3-2)9 = 9$.
\end{example}

By Proposition \ref{prop:genLD_properties}, we know that the $r$-th generalized Lee distances are in non-decreasing order. Therefore, for any $r \in \set{\sigma + 1, \ldots , s}$ the bound in Lemma \ref{lemma:bound_rth} is a valid upper bound for the minimum Lee distance of a code $\code$. However, as visible in Example \ref{ex:ell_n'}, the bounds on the $r$-th minimum Lee distances do not have to follow the same non-decreasing order. As they all hold as an upper bound to the minimum Lee distance of the code, the following bound is a direct consequence of Lemma \ref{lemma:bound_rth} by choosing the smallest among the bounds given in the statement.

\begin{corollary}\label{cor:SB_filtration}
    Given a code $\code \subseteq \zpsn$ of subtype $(k_0, \ldots, k_{s-1})$. For each $r \in \set{\sigma + 1, \ldots , s}$ let $\ell \geq 1$ and $(\ell, n')$ be the pair $(\ell^{(r-1)}, n'^{(r-1)})$ minimizing
    \begin{align*}
        p^{s - \ell^{(r-1)} + \sigma} + (n - K - n'^{(r-1)})M_{s - \ell^{(r-1)} + \sigma}.
    \end{align*}
    Then the codes minimum distance is bounded by
    \begin{align*}
        \LD(\code) \leq p^{s - \ell + \sigma} + (n - K - n')M_{s - \ell + \sigma}.
    \end{align*}
\end{corollary}

As for large rank $K$ this minimum can take a while to compute, we can also derive a slightly weaker bound depending on the maximal value $\ell_i$, which is easy to compute.
\begin{corollary}\label{cor:SB_filtration_2.0}
    Given a code $\code \subseteq \zpsn$ of subtype $(k_0, \ldots, k_{s-1})$ of maximal subtype $k_\sigma$. For each $r \in \set{\sigma + 1, \ldots , s}$ let $\ell := \max\set{\ell_i \st i = 1, \ldots , K}$ and define the corresponding value $n':= \max\set{n_{\ell_i} \st \ell_i = \ell, \, \text{for } i = 1, \ldots , K}$
    Then, the  minimum distance is bounded by
    \begin{align*}
        \LD(\code) \leq
        \begin{cases}
            p^{s - \ell + \sigma} + (n - K - n')M_{s - \ell + \sigma} & \text{if } \ell \geq 1, \\
            p^\sigma + (n-K)M_{\sigma} & \text{else}.
        \end{cases}      
    \end{align*}
\end{corollary}
In fact, we can identify  conditions, leading to four different cases for the bound provided in Corollary \ref{cor:SB_filtration_2.0}. 
For this very last observation, leading to the very last Lee-metric Singleton bound, we first need one last definition. Let $\mathcal{C} \subseteq \left( \mathbb{Z}/p^s\mathbb{Z}\right)^n$ be a linear code of maximal subtype $k_\sigma$ and assume that $\code_\sigma$ is generated by $(p^\sigma \mathbb{I} \ A)$. Let us denote the entries of $A$ as $a_{i,j},$ for $i \in \{1, \ldots, K\}$ and $j \in \{K+1, \ldots, n\}$. We define $$N'= \max\{ j \in \{K+1, \ldots, n\} \mid p \mid a_{i,j}, i \in \{1, \ldots, K\} \}.$$ That is $N'$ is  the maximal number of entries in a  row of $A$, which are divisible by $p.$

\begin{example}
    Let us consider the code over $\mathbb{Z}/27\mathbb{Z}$ generated by 
    $$G= \begin{pmatrix}
        1 & 0 & 3 & 6 \\ 
        0 & 1 & 18 & 1 
    \end{pmatrix}.$$
    The previous bound from Corollary \ref{cor:SB_filtration_2.0} would take $\ell=2$ and $n'=1,$ instead $N'=2$, as in the first row of $A$ we have two entries that are divisible by $p.$ In fact, this indicates  the minimum Hamming weight codeword in the socle, in this case 1.  Clearly, if $N'$ is large, it is beneficial to go until the socle.
\end{example}
\begin{enumerate}
    \item Case $\ell =  \sigma$ or $n'/2 \leq \frac{p^{s-\ell}-1}{p^{s -\sigma}-1}$. In this case we stay in $\code_{\sigma}$: 
        \begin{align*}
            \LD(\code) \leq p^{\sigma} + (n-K)M_\sigma
        \end{align*}
    \item Case $\ell =  s$. In this case we also stay in $\code_\sigma$, but observed some zero entries: 
        \begin{align*}
            \LD(\code) \leq p^{\sigma} + (n-k-n')M_{\sigma}
        \end{align*}
    \item Case $\ell \neq \sigma$ or $\ell \neq s$ and $n'/2 \geq \frac{p^{s-\ell}-1}{p^{s -\sigma}-1}$. In this case we can move to $\code_{s-\ell+\sigma}$:  
        \begin{align*}
            \LD(\code) \leq p^{s-\ell + \sigma} + (n-k-n')M_{s - \ell + \sigma}
        \end{align*}
    \item Case: if $n' \leq N' \frac{p^{\ell-\sigma}-p^{\ell-\sigma-1}}{p^{\ell-\sigma}-1}+(n-K-2)\frac{p^{\ell-\sigma-1}-1}{p^{\ell-\sigma}-1}$. In this case we go to the socle:
        \begin{align*}
            \LD(\code) \leq    p^{s - 1} + (n - K - N')M_{s - 1}. 
        \end{align*}
\end{enumerate}

Note also, that instead of taking the filtration subcodes $\mathcal{C}_i = \mathcal{C} \cap \langle p^i \rangle$, we could have also considered the torsion subcodes. 
\begin{definition}
    Let $\mathcal{C} \subseteq \left( \mathbb{Z} / p^s\mathbb{Z} \right)^n$. For $ i \in \{0, \ldots, s-1\}$, we call $\widetilde{\mathcal{C}_{i}}=\mathcal{C} \mod p^{s-i} \subseteq \left( \mathbb{Z}/p^{s-i} \mathbb{Z}\right)^n$ the $i$-th torsion code. 
\end{definition}
We can, however, immediately observe that the $i$-th torsion code represented as a code over the ambient space is naturally a subcode of the filtration subcode as $$p^i\widetilde{\mathcal{C}_i}\subseteq \mathcal{C}_i \subseteq \left( \mathbb{Z}/p^s\mathbb{Z}\right)^n,$$ with $\text{rk}(p^i \widetilde{\mathcal{C}_i})= \sum_{j=0}^{i-1} k_j < \text{rk}(\mathcal{C}_i) =K.$

In fact, any generator matrix of $\widetilde{\mathcal{C}_i}$  is a truncation of a generator matrix of $G$, i.e.,  we cut off the rows belonging to the subtypes $k_i, \ldots, k_{s-1}.$

Thus if we would define the $r$-th generalized Lee distances as $\LD^r(\mathcal{C})= \LD(\widetilde{\mathcal{C}_r})$, for $r \in \{0, \ldots, s-1\}$ then 
$\LD(\mathcal{C}) \leq \LD(\mathcal{C}_{i}) \leq \LD(p^{i} \tilde{\mathcal{C}_i})$. Thus, any upper bound on $\LD(p^{i} \tilde{\mathcal{C}_i})$ would serve as upper bound on $\LD(\mathcal{C}),$ but would be worse than taking directly bounds on the smaller $\LD(\code_i).$

Finally, we want to note, that the same considerations also apply to the Hamming metric. 
\begin{corollary} 
    Given a code $\code \subseteq \zpsn$ of subtype $(k_0, \ldots, k_{s-1})$ of maximal subtype $k_\sigma$. For each $r \in \set{\sigma + 1, \ldots , s}$ let $\ell := \max\set{\ell_i \st i = 1, \ldots , K}$ and define the corresponding value $n':= \max\set{n_{\ell_i} \st \ell_i = \ell, \, \text{for } i = 1, \ldots , K}$
    Then, the Hamming minimum distance is bounded by
    \begin{align*}
        \HD(\code) \leq
        \begin{cases}
            1 + (n - K - n') & \text{if } \ell \geq 1, \\
            1 + (n-K)  & \text{else}.
        \end{cases}      
    \end{align*}
\end{corollary} 

Note, that the Lee-metric version, that is Corollary \ref{cor:SB_filtration_2.0}, is not directly implied by the Hamming-metric bound. Such a direct bound would say
  \begin{align*}
        \LD(\code) \leq
        \begin{cases}
            M(1 + (n - K - n')) & \text{if } \ell \geq 1, \\
            M(1 + (n-K))  & \text{else}.
        \end{cases}      
    \end{align*}

    This is clearly a worse bound than our Lee-metric Singletonn bound of Corollary \ref{cor:SB_filtration_2.0}.
\subsection{Density of Optimal Codes with respect to Filtrations}\label{sec:MLD_filtration}\hfill\\
One interesting quantity is the number of codes of maximum achievable Lee distance for given parameters. We call a such a code a \textit{maximum Lee distance} (MLD) code. We have already seen that codes attaining the bounds based on the join-support and based on the column support are sparse as $p, s$ and $n$ tend to infinity.
In this subsection we discuss the density of MLD codes with respect  to the new Lee-metric Singleton bound \ref{cor:SB_filtration_2.0} from the filtration. 
If nothing else is stated we  consider a code $\code \subseteq \zpsn$ of rank $K$ and subtype $(k_0, \ldots, k_{s-1})$. 

Recall that the bound from Corollary \ref{cor:SB_filtration_2.0} is especially tight, if there are many zero positions in a row of a generator matric of a filtration subcode. Given the rank $K$ of a code $\code \subseteq \zpsn$ the probability that an entire row of $A$ is zero, where $A$ are the last $n-K$ columns of a generator matrix of a filtration $\code_{r-1}$ 
with $r \in \set{\sigma + 1, \ldots, s}$, is depending on $\sigma$, i.e., it depends on whether the code $\code$ is free or not.

For $n \tendsto \infty$ it is known \cite{free} that 
\begin{align*}
    \prob(\code \text{ is free}) =
    \begin{cases}
        1 & \text{if } R < 1/2, \\
        0 & \text{if } R > 1/2.
    \end{cases}
\end{align*}
Hence, in this case we would have to distinguish again the two cases. On the contrary for $p \tendsto \infty$, it is well-known that the code $\code$ is free with high probability, which implies that $\sigma = 0$. In this case, we have
\begin{enumerate}
    \item For every $i \in \set{1, \ldots, K}$, $\ell_i = 0$. Thus, the bound in Corollary \ref{cor:SB_filtration_2.0} can be reduced to 
    \begin{align*}
        \LD(\code) \leq 1 + (n-K)M,
    \end{align*}
    which coincides with the Singleton-like bound provided by \cite{shiromoto}.

    \item There is a $i \in \set{1, \ldots, K}$ with $\ell_i \neq 0$. In this case, we can find the pair $(\ell, n')$ as in Corollary \ref{cor:SB_filtration_2.0} and the minimum Lee distance is bounded by
    \begin{align*}
        \LD(\code) \leq p^{s-\ell} + (n-K-n')M_{s-\ell}.
    \end{align*}
\end{enumerate}

The following Lemma shows that for $p \tendsto \infty$ the first case occurs with high probability.
\begin{lemma}\label{lem:ell=0_growP}
    For a free linear code $\code \subseteq\zpsn$, as $p \tendsto \infty$, $\ell = 0$ with high probability.
\end{lemma}
\begin{proof}
    Note that $\prob(\ell = 0)$ is the probability that there is no multiple of $p$ contained in the last $n-K$ columns of a generator matrix $G$ of $\code$ in systematic form. More explicitly, it is the probability that all of the entries in the last $n-K$ columns of $G$ are units. That is,
    \begin{align*}
        \prob(\ell = 0) = \left(\frac{(p-1)p^{s-1}}{p^s}\right)^{K(n-K)} = \left(1 - \frac{1}{p}\right)^{K(n-K)}.
    \end{align*}
    Hence, letting $p$ grow to infinity and keeping $n$ and $K$ fixed, yields the desired result.
\end{proof}

This means that, with high probability, MLD codes are sparse as $p \tendsto \infty$, as codes attaining the bound \ref{shir} of Shiromoto are sparse.

Note that, letting $s$ grow to infinity and keeping $p$ fixed, the size of the ring $\zps$ still grows whereas the probability $\prob(\ell = 0)$ is a nonzero constant. This let us suggest, that codes attaining the bound on the minimum distance derived from filtration subcodes might not be sparse.\\

We start by discussing the case, where the code $\code$ is a free code, hence $\sigma = 0$. Free codes have a generator matrix of the form $(\mathbb{I}_K \st A)$, whit $A \in (\zps)^{K\times(n-K)}$.
If there is an $0<\widetilde{\ell}<s$ such that $n' = n_{\widetilde{\ell}} = n-K$, the filtration \smash{$\code_{s-\widetilde{\ell}}$} has an entire row equal to zero. This results in having an \smash{$(s-\widetilde{\ell})$}-th generalized Lee distance of \smash{$p^{s - \widetilde{\ell}}$} and hence $\LD(\code) \leq p^{s - \widetilde{\ell}}$.\\
Let us investigate on the probability for $A$ having a maximal $0<\ell_i = \widetilde{\ell} < s$ with corresponding $n' = n-K$. This requires that all other rows of $A$ are contained at most in the ideal $\langle p^{\widetilde{\ell}} \rangle$. The probability that $A$ is of this form is therefore
\begin{align*}
    \mathcal{P} &:= \frac{(p^{s - \widetilde{\ell}}- p^{s - \widetilde{\ell}-1})^{(n-K)} (p^{s-1} - p^{s-\widetilde{\ell}-1})^{(K-1)} (p^{s - \widetilde{\ell}} - p^{s - \widetilde{\ell}-1})^{(n-K-1)(K-1)} }{(p^s)^{(n-K)K}}\\
    &= (p^{ - \widetilde{\ell}}- p^{ - \widetilde{\ell}-1})^{(n-K)} (p^{s-1} - p^{-\widetilde{\ell}-1})^{(K-1)} (p^{ - \widetilde{\ell}} - p^{ - \widetilde{\ell}-1})^{(n-K-1)(K-1)}\\
    &= \left( \frac{1}{p^{\widetilde{\ell}}} -  \frac{1}{p^{\widetilde{\ell}+1}} \right)^{(n-K-1)K + 1} \left( \frac{1}{p} -  \frac{1}{p^{\widetilde{\ell}+1}} \right)^{(K-1)}.
\end{align*}
This probability tends to zero as $n \tendsto \infty$, and thus MLD codes are sparse with respect to the bound given in Corollary \ref{cor:SB_filtration_2.0} and $n \tendsto \infty$. However, since $\mathcal{P}$ does not dependent on $s$, as $s \tendsto \infty$ and is a nonzero constant, this  implies neither sparsity nor density. In any case, we have that with a probability $\prob(\code \text{ is free}) \mathcal{P}$ the minimum distance of the code is bounded by $\LD(\code)\leq p^{s-\widetilde{\ell}}$.\\

Let us consider now codes that achieve the bound on the minimum Lee distance based on filtration subcodes, i.e., Corollary \ref{cor:SB_filtration_2.0} and check whether this fixes the $r$-th generalizes Lee distances.

Clearly, if $\mathcal{C}$ has maximal subtype $k_\sigma$ and attains the bound in Corollary \ref{cor:SB_filtration_2.0}, then 
$\LD(\code)= \LD^1(\code) = \cdots = \LD^{\sigma-1}(\code)= \LD(\code_\sigma).$
If $\sigma=s-1$, or we are in the case 4, i.e., $n' \leq N' \frac{p^{\ell-\sigma}-p^{\ell-\sigma-1}}{p^{\ell-\sigma}-1}+(n-K-2)\frac{p^{\ell-\sigma-1}-1}{p^{\ell-\sigma}-1}$. In this case we go to the socle, and hence all $\LD^r(\code)$ are equal. 
If we are not in case 4, the behaviour  of the filtration subcodes $\code_r$ with $r\geq \sigma$ is more unpredictable. 

As already discussed above there are codes with several properties which are attaining the bound in Corollary \ref{cor:SB_filtration_2.0}.

One class of codes that we want to consider are those having $n' = n-K$. Assuming that such a code attains the bound, the following result gives us a closed expression for the $r$-th generalized Lee distances for all $r$.
\begin{proposition}
    Let $\code \subseteq \zpsn$ of rank $K$, subtype $(k_0, \ldots, k_s)$ of maximal subtype $k_\sigma$, and tuple $(\ell, n-K)$, such that $\LD(\code) = \LD^{s-\ell+\sigma}(\code)$. Then the $r$-th generalized Lee distance is given by
    \begin{align*}
        \LD^r(\code) =
        \begin{cases}
            p^{s-\ell+\sigma} & \text{for every } r \leq s - \ell + \sigma, \\
            p^r & \text{for every } r > s - \ell + \sigma. 
        \end{cases}
    \end{align*}
\end{proposition}
\begin{proof}
    Since $\LD(\code) = \LD^{s-\ell+\sigma}(\code)$ and since the $r$-th generalized Lee distances are increasing in $r$, we have $\LD^r(\code) = \LD^{s-\ell+\sigma}(\code)$ for every $r \leq s - \ell + \sigma$. Hence, the first case is clear.
    For the second case note that $\code_{s - \ell + \sigma}$ admits a generator matrix containing only zeros in the last $n-K$ columns. These entries remain to be zero for every filtration $\code_r$ with $r > {s - \ell + \sigma}$. Hence, the minimum distance $\LD(\code_r)$ is always given by $p^r$.
\end{proof}

\subsection{Invariance under Isometry in the Lee Metric}\label{sec:invar_filtration}\hfill\\
Finally, we observe again that the $r$-th generalized Lee distance for a code $\code \subseteq \zpsn$ coincides with the $r$-th generalized Lee distance of a code $\code' \subseteq \zpsn$ that is equivalent to $\code$. 
\begin{proposition}
    Let $\code \subseteq \zpsn$ of rank $K$ and let $\code' \subseteq \zpsn$ another code that is equivalent to $\code$. Then, for every $r \in \set{1, \ldots , K}$, we have
    $$\LD^r(\code) = \LD^r(\code').$$
\end{proposition}
\begin{proof}
    Let $\phi$ denote an isometry preserving the Lee distance. Recall that this isometry can only consist of permutations and multiplications by $\pm 1$. Furthermore, recall that the $r$-th generalized Lee distance is given by the minimum Lee distance of the $r$-th filtration subcode $\code_r$ of $\code$, i.e.,
    \begin{align*}
        \LD^r(\code) = \LD(\code_r).
    \end{align*}
    By the inclusion property of the filtrations, we have $\code_r \subseteq \code$. Note that $\phi$ additionally preserves this inclusion property, i.e.,
    \begin{align*}
        \code'_r := \phi(\code_r) \subseteq \phi(\code).
    \end{align*}
    Hence, the minimum Lee distances of $\code_r$ and $\code'_r$ coincide.
\end{proof}

\section{Comparison of the Bounds}\label{sec:comparison}
At this point let us compare the bound of Corollary \ref{cor:SB_filtration_2.0} to the bounds derived from the new puncturing argument (Theorem \ref{thm:SBpunct}), the join-Lee support (Theorem \ref{thm:SB_joinsupp}), to the column support (Corollary \ref{cor:SB_column_cancelling}) and to the bounds provided by \cite{shiromoto,alderson}. We  do so by providing first some examples that attain the bound from Corollary \ref{cor:SB_filtration_2.0} and compare it to the other bounds.
\begin{example}
\begin{enumerate}
    \item 
    Let $\code \subseteq (\Zp{9})^4$ generated by 
    \begin{align*}
        G = 
        \begin{pmatrix}
            1 & 0 & 0 & 2\\
            0 & 1 & 0 & 6\\
            0 & 0 & 1 & 4
        \end{pmatrix}.
    \end{align*}
    We quickly observe that this code has a minimum Lee distance $\LD(\code) = 3$. For the last $n-K = 1$ column, we note, that all the entries live in the ideal generated by $1$. This means that $\ell = 0$ and $n' = n-K = 1$. Then the bounds are computed as follows.
    \begin{align*}
        \text{Filtration:}\quad &\LD(\code) \leq 3 &&\hspace{-20mm}(\text{Corollary \ref{cor:SB_filtration_2.0}})\\
        \text{Join-support:}\quad &\LD(\code) \leq 6 &&\hspace{-20mm}(\text{Theorem \ref{thm:SB_joinsupp}}) \\
        \text{Column support:}\quad &\LD(\code) \leq 5 &&\hspace{-20mm}(\text{Corollary \ref{cor:SB_column_cancelling}})\\
        \text{New puncturing:}\quad &\LD(\code) \leq 8 &&\hspace{-20mm}(\text{Theorem \ref{thm:SBpunct}})\\
        \text{Shiromoto:}\quad &\LD(\code) \leq 5 &&\hspace{-20mm}(\text{\cite{shiromoto}})\\
        \text{Alderson - Huntemann:}\quad &\LD(\code) \leq 4 &&\hspace{-20mm}(\text{\cite{alderson}})\\
    \end{align*}

    \item 
    Let $\code \subseteq (\Zp{27})^5$ generated by 
    \begin{align*}
        G = 
        \begin{pmatrix}
            1 & 10 & 4 & 20 & 9\\
            0 & 3 & 9 & 18 & 9
        \end{pmatrix}.
    \end{align*}
    The minimum Lee distance of this code is $\LD(\code) = 9$. For the last $n-K = 3$ columns, we quickly compute $\ell'=2$ and $n'=1$. Then the bounds are computed as follows.
    \begin{align*}
        \text{Filtration:}\quad &\LD(\code) \leq 9 &&\hspace{-20mm}(\text{Corollary \ref{cor:SB_filtration}})\\
        \text{Filtration:}\quad &\LD(\code) \leq 9 &&\hspace{-20mm}(\text{Corollary \ref{cor:SB_filtration_2.0}})\\
        \text{Join-support:}\quad &\LD(\code) \leq 36 &&\hspace{-20mm}(\text{Theorem \ref{thm:SB_joinsupp}}) \\
        \text{Column support:}\quad &\LD(\code) \leq 38 &&\hspace{-20mm}(\text{Corollary \ref{cor:SB_column_cancelling}})\\
        \text{New puncturing:}\quad &\LD(\code) \leq 48 &&\hspace{-20mm}(\text{Theorem \ref{thm:SBpunct}})\\
        \text{Shiromoto:}\quad &\LD(\code) \leq 40 &&\hspace{-20mm}(\text{\cite{shiromoto}})\\
        \text{Alderson - Huntemann:}\quad &\text{not existing} &&\hspace{-20mm}(\text{\cite{alderson}})
    \end{align*}

    \item In this example let us consider the code $\code \subseteq (\Zp{125})^6$ generated by
    \begin{align*}
        G = 
        \begin{pmatrix}
            1 & 0 & 25 & 50 & 75 & 100 \\
            0 & 1 & 2 & 3 & 4 & 5
        \end{pmatrix}.
    \end{align*}
    This code has minimum distance $\LD(\code) = 5$.
    Note that the two bounds with respect to the filtration (Corollary \ref{cor:SB_filtration} and \ref{cor:SB_filtration_2.0}) coincide. Hence, we obtain
    \begin{align*}
        \text{Filtration:}\quad &\LD(\code) \leq 5 &&\hspace{-10mm}(\text{Corollary \ref{cor:SB_filtration} and \ref{cor:SB_filtration_2.0}})\\
        \text{Join-support:}\quad &\LD(\code) \leq 200 &&\hspace{-10mm}(\text{Theorem \ref{thm:SB_joinsupp}}) \\
        \text{Column support:}\quad &\LD(\code) \leq 247 &&\hspace{-10mm}(\text{Corollary \ref{cor:SB_column_cancelling}})\\
        \text{New puncturing:}\quad &\LD(\code) \leq 300 &&\hspace{-10mm}(\text{Theorem \ref{thm:SBpunct}})\\
        \text{Shiromoto:}\quad &\LD(\code) \leq 249 &&\hspace{-10mm}(\text{\cite{shiromoto}})\\
        \text{Alderson - Huntemann:}\quad &\LD(\code) \leq 248 &&\hspace{-10mm}(\text{\cite{alderson}})
    \end{align*}
\end{enumerate}

\end{example}
We now compare the bounds for different parameters. We will leave out the bound given by the column support, i.e., Corollary \ref{cor:SB_column_cancelling}, as we would need to consider too many different parameters which would not fit in the overview.
\begin{longtable}{|p{2cm}||p{2.5cm}|p{2.5cm}|p{2.5cm}|p{2.5cm}|}
    \hline
    $(n, K, p^s, \sigma)$& Alderson - Huntemann  ~~\cite{alderson} & Shiromoto \cite{shiromoto} & Join-support (Theorem \ref{thm:SB_joinsupp})& Filtration (Corollary \ref{cor:SB_filtration_2.0}) $(\ell, n')$\\ 
    \hline\hline 
    \multirow{7}{*}{$(6, 3, 9, 0)$}& \multirow{7}{*}{$12$} & \multirow{7}{*}{$16$}   & \multirow{7}{*}{$12$}  & $(0, 3): \; 13$\\ 
    & & & & $(1, 1): \; 9$\\
    & & & & $(1, 2): \; 6$\\
    & & & & $(1, 3): \; 3$\\
    & & & & $(2, 1): \; 9$\\
    & & & & $(2, 2): \; 5$\\
    & & & & $(2, 3): \; 1$\\
    \hline
    \multirow{4}{*}{$(6, 3, 9, 1)$}& \multirow{4}{*}{Not existing} & \multirow{4}{*}{$16$}   & \multirow{4}{*}{$12$}  & $(1, \star): \; 12$\\ 
    & & & & $(2, 1): \; 9$\\
    & & & & $(2, 2): \; 6$\\
    & & & & $(2, 3): \; 3$\\
    \hline
    \multirow{10}{*}{$(6, 3, 125, 0)$}& \multirow{10}{*}{$186$} & \multirow{10}{*}{$248$}   & \multirow{10}{*}{$200$}  & $(0, 3):\;187$\\ 
    & & & & $(1, 1): \; 125$\\
    & & & & $(1, 2): \; 75$\\
    & & & & $(1, 3): \; 25$\\
    & & & & $(2, 1): \; 125$\\
    & & & & $(2, 2): \; 65$\\
    & & & & $(2, 3): \; 5$\\
    & & & & $(3, 1): \; 125$\\
    & & & & $(3, 2): \; 63$\\
    & & & & $(3, 3): \; 1$\\
    \hline
    \multirow{7}{*}{$(6, 3, 125, 1)$}&   & \multirow{7}{*}{$248$}   & \multirow{7}{*}{$200$}  & $(1, \star):\;185$\\ 
    & & & & $(2, 1): \; 125$\\
    & $248$ & & & $(2, 2): \; 75$\\
    & \tiny (only for subtype & & & $(2, 3): \; 2$\\
    & \tiny $(0, 3, 0)$) & & & $(3, 1): \; 125$\\
    & & & & $(3, 2): \; 65$\\
    & & & & $(3, 3): \; 5$\\
    \hline
    \multirow{4}{*}{$(6, 3, 125, 2)$}& $310$ \tiny(only for sub- & \multirow{4}{*}{$248$}   & \multirow{4}{*}{$200$}  & $(2, \star):\;175$\\ 
    & \tiny type $(0, 0, 3)$) & & & $(3, 1): \; 125$\\
    & $248$ \tiny(only for sub- & & & $(3, 2): \; 75$\\
    & \tiny type $(1, 1, 1)$) & & & $(3, 3): \; 25$\\
    \hline
    \caption{Comparison of the bounds on the minimum Lee distance of a code with given parameters. \label{tab:compare_free}}
\end{longtable}
Let us focus first on a free code, i.e., $\sigma = 0$). Observe, that if the last $n-K$ columns of a generator matrix consist only of nonunits, i.e., $\ell = 0$) the bound by Alderson and Huntemann beats our bounds. However, as soon as $\ell \neq 0$ the new  bound based on the minimum distance of filtration subcodes (Corollary \ref{cor:SB_filtration_2.0}) always outperforms any other bound. In Table \ref{tab:compare_free} we also observe, that the bound provided by Shiromoto is the loosest.

For nonfree codes, recall that the bound in \cite{alderson} only works for integer $\mathbb{Z}/p^s\mathbb{Z}$-dimensions $k> 1$. 
Furthermore, we note that for a given $\sigma \geq 1$ we always have $\ell \geq \sigma$ and if $\ell = \sigma$ the filtration bound (Corollary \ref{cor:SB_filtration_2.0}) is the same for any  $n'$. This is denoted by $n' = \star$ in Table \ref{tab:compare_free}. In any of the parameters presented, the bound based on the minimum Lee distance of a filtration subcodes of the code  (Corollary \ref{cor:SB_filtration_2.0}) outperforms all other bounds.

\section{Conclusions and Future Work}\label{sec:conclusions}

In this paper we presented several novel definitions of Lee support and the corresponding generalized Lee weights of subcodes of a fixed rank. 
These give raise to new Lee-metric Singleton bounds, that beat the previous bound by Shiromoto \cite{shiromoto}, which follows from a puncturing argument. 
However, their optimal codes are still sparse for $s,n$ or $p$ going to infinity. This led us to consider different subcodes, namely the  filtration subcodes. Bounding their minimum distances gives a sharper Singleton bound, that finally has the desired property; their optimal codes are not sparse for $s$ going to infinity.

The open question remains, whether there exists a bound on the minimum Lee distance of codes, which is such that their optimal codes are dense for $s,n$ or $p$ going to infinity.

\section*{Acknowledgements}\label{sec:ack} 
Violetta Weger is  supported   by the European Union's Horizon 2020 research and innovation programme under the Marie Sk\l{}odowska-Curie grant agreement no. 899987.

\bibliographystyle{plain}

\bibliography{References}

\begin{thebibliography}{10}

\bibitem{alderson}
Tim~L Alderson and Svenja Huntemann.
\newblock On maximum {L}ee distance codes.
\newblock {\em Journal of Discrete Mathematics}, 2013, 2013.

\bibitem{heide2}
Jared Antrobus and Heide Gluesing-Luerssen.
\newblock Maximal {F}errers diagram codes: constructions and genericity
  considerations.
\newblock {\em IEEE Transactions on Information Theory}, 65(10):6204--6223,
  2019.

\bibitem{forney}
Alexander Barg and G~David Forney.
\newblock Random codes: Minimum distances and error exponents.
\newblock {\em IEEE Transactions on Information Theory}, 48(9):2568--2573,
  2002.

\bibitem{free}
Eimear Byrne, Anna-Lena Horlemann, Karan Khathuria, and Violetta Weger.
\newblock Density of free modules over finite chain rings.
\newblock {\em Linear Algebra and its Applications}, 651:1--25, 2022.

\bibitem{BR19}
Eimear Byrne and Alberto Ravagnani.
\newblock Partition-balanced families of codes and asymptotic enumeration in
  coding theory.
\newblock {\em Journal of Combinatorial Theory, Series A}, 171:105169, 2020.

\bibitem{byrne2023bounds}
Eimear Byrne and Violetta Weger.
\newblock Bounds in the lee metric and optimal codes.
\newblock {\em Finite Fields and Their Applications}, 87:102151, 2023.

\bibitem{cardell2020generalized}
Sara~D Cardell, Marcelo Firer, and Diego Napp.
\newblock Generalized column distances.
\newblock {\em IEEE Transactions on Information Theory}, 66(11):6863--6871,
  2020.

\bibitem{delsarte}
Philippe Delsarte.
\newblock Bilinear forms over a finite field, with applications to coding
  theory.
\newblock {\em Journal of combinatorial theory, Series A}, 25(3):226--241,
  1978.

\bibitem{doughertybook}
Steven~T Dougherty.
\newblock {\em Algebraic coding theory over finite commutative rings}.
\newblock Springer, 2017.

\bibitem{dougherty2002generalized}
Steven~T Dougherty, M~Gupta, and Keisuke Shiromoto.
\newblock On generalized weights for codes over finite rings.
\newblock {\em preprint}, 2002.

\bibitem{douandshi}
Steven~T Dougherty and Keisuke Shiromoto.
\newblock {MDR} codes over $\mathbb{Z}_k$.
\newblock {\em IEEE Transactions on Information Theory}, 46(1):265--269, 2000.

\bibitem{gabidulin}
Ernest~Mukhamedovich Gabidulin.
\newblock Theory of codes with maximum rank distance.
\newblock {\em Problemy peredachi informatsii}, 21(1):3--16, 1985.

\bibitem{heide}
Heide Gluesing-Luerssen.
\newblock On the sparseness of certain linear {MRD} codes.
\newblock {\em Linear Algebra and its Applications}, 596:145--168, 2020.

\bibitem{gorla2022generalized}
Elisa Gorla and Alberto Ravagnani.
\newblock Generalized weights of codes over rings and invariants of monomial
  ideals.
\newblock {\em arXiv preprint arXiv:2201.05813}, 2022.

\bibitem{gorla2022generalizedcolumn}
Elisa Gorla and Flavio Salizzoni.
\newblock Generalized column distances, 2022.

\bibitem{gorla2023generalized}
Elisa Gorla and Flavio Salizzoni.
\newblock Generalized weights of convolutional codes.
\newblock {\em IEEE Transactions on Information Theory}, 2023.

\bibitem{anina}
Anina Gruica and Alberto Ravagnani.
\newblock Common complements of linear subspaces and the sparseness of mrd
  codes.
\newblock {\em SIAM Journal on Applied Algebra and Geometry}, 6(2):79--110,
  2022.

\bibitem{hammons1994z}
A~Roger Hammons, P~Vijay Kumar, A~Robert Calderbank, Neil~JA Sloane, and
  Patrick Sol{\'e}.
\newblock The $z_4$-linearity of {K}erdock, {P}reparata, {G}oethals, and
  related codes.
\newblock {\em IEEE Transactions on Information Theory}, 40(2):301--319, 1994.

\bibitem{gen}
Tor Helleseth, Torleiv Kl{\o}ve, and Johannes Mykkeltveit.
\newblock The weight distribution of irreducible cyclic codes with block
  lengths $n_1 ((q^\ell- 1)/n$).
\newblock {\em Discrete Mathematics}, 18(2):179--211, 1977.

\bibitem{leeZ4}
Anna-Lena Horlemann-Trautmann and Violetta Weger.
\newblock Information set decoding in the {L}ee metric with applications to
  cryptography.
\newblock {\em Advances in Mathematics of Communications}, 15(4), 2021.

\bibitem{komamiya}
Yasuo Komamiya.
\newblock Application of logical mathematics to information theory.
\newblock {\em Proc. 3rd Japan. Nat. Cong. Appl. Math}, 437:3, 1953.

\bibitem{RelativeRankweight}
Jun Kurihara, Ryutaroh Matsumoto, and Tomohiko Uyematsu.
\newblock Relative generalized rank weight of linear codes and its applications
  to network coding.
\newblock {\em IEEE Transactions on Information Theory}, 61(7):3912--3936,
  2015.

\bibitem{rankGenWeight}
Jun Kurihara, Tomohiko Uyematsu, and Ryutaroh Matsumoto.
\newblock New parameters of linear codes expressing security performance of
  universal secure network coding.
\newblock In {\em 2012 50th Annual Allerton Conference on Communication,
  Control, and Computing (Allerton)}, pages 533--540, 2012.

\bibitem{lee1958some}
C~Lee.
\newblock Some properties of nonbinary error-correcting codes.
\newblock {\em IRE Transactions on Information Theory}, 4(2):77--82, 1958.

\bibitem{loidreau}
Pierre Loidreau.
\newblock Asymptotic behaviour of codes in rank metric over finite fields.
\newblock {\em Designs, codes and cryptography}, 71:105--118, 2014.

\bibitem{ale}
Alessandro Neri, Anna-Lena Horlemann-Trautmann, Tovohery Randrianarisoa, and
  Joachim Rosenthal.
\newblock On the genericity of maximum rank distance and {G}abidulin codes.
\newblock {\em Designs, Codes and Cryptography}, 86(2):341--363, 2018.

\bibitem{ravagnani2016generalized}
Alberto Ravagnani.
\newblock Generalized weights: an anticode approach.
\newblock {\em Journal of Pure and Applied Algebra}, 220(5):1946--1962, 2016.

\bibitem{fuleeca}
Stefan Ritterhoff, Georg Maringer, Sebastian Bitzer, Violetta Weger, Patrick
  Karl, Thomas Schamberger, Jonas Schupp, and Antonia Wachter-Zeh.
\newblock Fu{L}eeca: A {L}ee-based signature scheme.
\newblock {\em Cryptology ePrint Archive}, 2023.

\bibitem{roth}
Ron~M Roth.
\newblock Maximum-rank array codes and their application to crisscross error
  correction.
\newblock {\em IEEE transactions on Information Theory}, 37(2):328--336, 1991.

\bibitem{segre}
Beniamino Segre.
\newblock Curve razionali normali e $k$-archi negli spazi finiti.
\newblock {\em Annali di Matematica Pura ed Applicata}, 39(1):357--379, 1955.

\bibitem{shiromoto}
Keisuke Shiromoto.
\newblock Singleton bounds for codes over finite rings.
\newblock {\em Journal of Algebraic Combinatorics}, 12(1):95--99, 2000.

\bibitem{singleton}
Richard Singleton.
\newblock Maximum distance $q$-nary codes.
\newblock {\em IEEE Transactions on Information Theory}, 10(2):116--118, 1964.

\bibitem{LeeNP}
Violetta Weger, Karan Khathuria, Anna-Lena Horlemann, Massimo Battaglioni,
  Paolo Santini, and Edoardo Persichetti.
\newblock On the hardness of the {L}ee syndrome decoding problem.
\newblock {\em Advances in Mathematics of Communications}, pages 0--0, 2022.

\bibitem{wei1991generalized}
Victor~K Wei.
\newblock Generalized hamming weights for linear codes.
\newblock {\em IEEE Transactions on information theory}, 37(5):1412--1418,
  1991.

\bibitem{wood}
Jay~A Wood.
\newblock The structure of linear codes of constant weight.
\newblock {\em Trans. American Math. Soc.}, 354(3):1007--1026, 2001.

\end{thebibliography}

\end{document}